\documentclass[%
reprint,
superscriptaddress,
showpacs,preprintnumbers,
 amsmath,
 amssymb,
aps,
pre,
floatfix,
]{revtex4-1}

\usepackage{microtype}
\usepackage{graphicx}
\usepackage{dcolumn}
\usepackage{bm}
\usepackage[colorlinks,allcolors=blue]{hyperref}

\newcommand{\dpar}[1]{\left(#1\right)}

\usepackage{amsmath,amssymb,amsthm}

\newtheorem*{thm}{Theorem}
\begin{document}
\title{Characterizing Time Series via Complexity-Entropy Curves}

\author{Haroldo V. Ribeiro}\email{hvr@dfi.uem.br}
\author{Max Jauregui}
\affiliation{Departamento de F\'isica, Universidade Estadual de Maring\'a, Maring\'a, PR 87020-900, Brazil}
\author{Luciano Zunino}
\affiliation{Centro de Investigaciones \'Opticas (CONICET La Plata - CIC), C.C. 3, 1897 Gonnet, Argentina}
\affiliation{Departamento de Ciencias B\'asicas, Facultad de Ingenier\'ia, Universidad Nacional de La Plata (UNLP), 1900 La Plata, Argentina}
\author{Ervin K. Lenzi}
\affiliation{Departamento de F\'isica, Universidade Estadual de Ponta Grossa, Ponta Grossa, PR 84030-900, Brazil}
\date{\today}

\begin{abstract}
The search for patterns in time series is a very common task when dealing with complex systems. This is usually accomplished by employing a complexity measure such as entropies and fractal dimensions. However, such measures usually only capture a single aspect of the system dynamics. Here, we propose a family of complexity measures for time series based on a generalization of the complexity-entropy causality plane. By replacing the Shannon entropy by a mono-parametric entropy (Tsallis $q$-entropy) and after considering the proper generalization of the statistical complexity ($q$-complexity), we build up a parametric curve (the $q$-complexity-entropy curve) that is used for characterizing/classifying time series. Based on simple exact results and numerical simulations of stochastic processes, we show that these curves can distinguish among different long-range, short-range and oscillating correlated behaviors. Also, we verify that simulated chaotic and stochastic time series can be distinguished based on whether these curves are open or closed. We further test this technique in experimental scenarios related to chaotic laser intensity, stock price, sunspot, and geomagnetic dynamics, confirming its usefulness. Finally, we prove that these curves enhance the automatic classification of time series with long-range correlations and interbeat intervals of healthy subjects and patients with heart disease.
\end{abstract}

\maketitle
\section{Introduction}
The study of complex systems often shares the goal of analyzing empirical time series aiming to extract patterns or laws that rule the system dynamics. In order to perform this task, it is very common to employ a complexity measure such as algorithmic complexity~\cite{Kolmogorov1965}, entropies~\cite{Shannon1948}, relative entropies~\cite{KullbackLeibler1951}, fractal dimensions~\cite{Mandelbrot}, and Lyapunov exponents~\cite{Lyapunov}. Researchers have actually defined several complexity measures, a fact directly related to the difficulty of accurately defining the meaning of complexity. However, the majority of the available measures depends on specific algorithms and tuning parameters, which usually creates great difficulties for research reproducibility. 

To overcome this problem, Bandt and Pompe~\cite{BandtPompe2002} have introduced a complexity measure based on the comparison of neighboring values that can be easily applied to any time series. The application of this technique (the permutation entropy) is widely spread over the scientific community~\cite{PhysRevLett.116.033902,Lin2016128,PhysRevE.91.023101,PhysRevE.89.012905,PhysRevE.85.021906,srep07784,srep01778} mainly because of its simplicity and ability to distinguish among regular, chaotic and random time series. Regarding the particular issue of distinguishing between chaotic and stochastic processes, Rosso~\textit{et al.}~\cite{RossoLarrondoMartinPlastinoFuentes2007} have shown that the permutation entropy alone is not enough for accomplishing this task. They observed, for instance, that the value of the permutation entropy calculated for the logistic map at fully developed chaos is very close to the value obtained for long-range correlated noises. Because of that, Rosso~\textit{et al.}~\cite{RossoLarrondoMartinPlastinoFuentes2007} have employed the ideas of Bandt and Pompe together with a diagram proposed by L\'opez-Ruiz \textit{et al.}~\cite{LopezManciniCalbet1995}. This diagram is composed of the values of a relative entropic measure (the statistical complexity) versus the Shannon entropy, both calculated within the framework of Bandt and Pompe. Rosso~\textit{et al.} have named this diagram as complexity-entropy causality plane, and by using it they were able to distinguish between several time series of stochastic and chaotic nature. The causality plane has also proved its usefulness in several applications~\cite{PhysRevE.91.023101,PhysRevE.89.012905,srep07784,Jovanovic2016,Stosic20161136,ribeiro2012complexity} and has been generalized for considering higher dimensional data~\cite{Ribeiro_etal2012}, different time~\cite{PhysRevE.86.046210} and spatial resolutions~\cite{Zunino2016679}. 

Here, we propose to extend the causality plane for considering a mono-parametric entropy in replacement of the Shannon entropy. In particular, we have considered the Tsallis $q$-entropy~\cite{Tsallis1988,Tsallis} (that recovers the Shannon entropy for $q=1$) together with the proper generalization of the statistical complexity~\cite{MartinPlastinoRosso2006} ($q$-complexity). The values of the parameter $q$ in the Tsallis entropy give different weights to the underlying probabilities of the system, accessing different dynamical scales and producing a family of complexity measures that capture some of the different meanings of complexity. Moreover, the Tsallis $q$-entropy has already proved to be useful for enhancing the performance of computational techniques such as in optimization problems~\cite{StarioloTsallis1995,TsallisStariolo1996,AndricioaeiStraub1997} and image thresholding~\cite{Albuquerque_etal2004,JalabIbrahimAhmed2016} as well as has been previously implemented for characterizing fractal stochastic processes~\cite{zunino2008fractional}.


Thus, for a given time series, we build up a parametric curve composed of the values of the $q$-complexity versus the $q$-entropy, which we will call the \emph{$q$-complexity-entropy curve}. Based on simple exact results, we discuss some general properties of these curves, and next we present an exhaustive list of applications based on numerical simulations and empirical data. These applications show that the $q$-complexity-entropy curve can capture dynamical aspects of time series that are not properly identified only by the point for $q=1$, which corresponds to the complexity-entropy causality plane of Rosso~\textit{et al.}. The rest of this article is organized as follows. Section~2 is devoted for reviewing the Bandt and Pompe approach and the complexity-entropy causality plane of Rosso~\textit{et al.}. Also in this section, we present our generalization for considering the Tsallis $q$-entropy as well as some general properties of the $q$-complexity-entropy curve. Section~3 presents our numerical experiments with time series from the fractional Brownian motion, harmonic noise, and chaotic maps. In Section~4, we discuss some real world applications involving time series from laser dynamics, sunspot numbers, stock prices, human heart rate, and Earth's magnetic activity. Section~5 ends this paper with some concluding remarks.

\section{Generalized entropy and complexity measures within the Bandt and Pompe framework}
We start by reviewing the approach of Bandt and Pompe~\cite{BandtPompe2002} for extracting the probabilities related to the ordinal dynamics of the elements of a time series. For a given time series $\{x_i\}_{i=1,\dots,n}$, we construct $(n-d+1)$ overlapping partitions of length $d>1$ represented by
\begin{equation}
s\to\{x_{s-(d-1)},x_{s-(d-2)},\dots,x_{s}\}\,,
\end{equation}
where $s=d,d+1,\dots,n$. For each $s$, we evaluate the permutations $\pi_j=\{r_0,r_1,\ldots,r_{d-1}\}$ of $\{0,1,\dots,d-1\}$ defined by the ordering $x_{s-r_{d-1}}\le x_{s-r_{d-2}}\le\ldots\le x_{s-r_{0}}$, and we associate to each permutation $\pi_j$ (with $j=1,\ldots,d!$) the probability
\begin{equation}
p_j(\pi_j) = \frac{\text{the number of}~s~\text{that has type}~\pi_j}{n-d+1}\,.
\end{equation}
The components of the probability distribution \mbox{$P=\{p_j(\pi_j)\}_{j=1,\ldots,d!}$} represent the odds of finding a segment of length $d>1$ within the time series in a given order. For instance, for the time series $\{2,4,3,5\}$, we can create tree partitions of size $d=2$: $(s=2)\to\{2,4\}$, $(s=3)\to\{4,3\}$, and $(s=4)\to\{3,5\}$. For each one, we associate the permutations $\{0,1\}$, $\{1,0\}$ and $\{0,1\}$, respectively; consequently, the probability distribution is $P=\{2/3,1/3\}$. Thus, the probability distribution $P=\{p_j(\pi_j)\}_{j=1,\ldots,d!}$ provides information about the ordering dynamics for a given time scale defined by the value of $d$, often called the embedding dimension. In the Bandt and Pompe framework, $d$ is a parameter whose value must satisfy the condition $n\gg d!$ in order to obtain reliable statistics for all $d\,!$ possible permutations occurring in the time series.

Given the probability distribution $P=\{p_j(\pi_j)\}_{j=1,\ldots,d!}$, Bandt and Pompe proposed to employ the normalized Shannon entropy
\begin{equation}\label{eq:shannon}
H_1(P) = \frac{S_{1}(P)}{S_{1}(U)}\,,
\end{equation}
where $S_1(P)=\sum_{j=1}^{d!} p_j \log \frac{1}{p_j}$ is the Shannon entropy and $U=\{1/d\,!\}_{j=1, 2, \dots d!}$ is the uniform distribution (so $S_{1}(U) = \log d\,!$),
as a natural measure of complexity. By following Bandt and Pompe's idea together with the diagram of L\'opez-Ruiz \textit{et al.}~\cite{LopezManciniCalbet1995}, Rosso~\textit{et al.}~\cite{RossoLarrondoMartinPlastinoFuentes2007} have proposed to further calculate a second complexity measure defined by
\begin{equation}\label{eq:statcomplexity}
C_1(P)=\frac{D_1(P,U) H_1(P)}{D_1^*}\,,
\end{equation}
where $D_1(P,U)$ is a relative entropic measure (the Jensen-Shannon divergence) between the empirical distribution $P=\{p_j(\pi_j)\}_{j=1,\ldots,d!}$ and the uniform distribution $U=\{1/d!\}_{j=1,\ldots,d!}$. This relative measure can be defined in terms of the symmetrized Kullback-Leibler divergence ($K_1(P|R)=-\sum p_i \log r_i/p_i$, with $P$ and $R$ probability distributions) and is written as
\begin{equation}
\begin{split}
D_1(P,U) &= \frac{1}{2}K_1\left(P\biggl|\frac{P+U}{2}\right) + \frac{1}{2}K_1\left(U\biggl|\frac{P+U}{2}\right)\\
&=\left[ S_1\left(\frac{P+U}{2}\right) - \frac{S_1(P)}{2} - \frac{S_1(U)}{2} \right]\,,
\end{split}
\end{equation}
with $\frac{P+U}{2} = \{\frac{p_j(\pi_j)+(1/d!)}{2}\}_{j=1,\ldots,d!}$; while
\begin{equation}
\begin{split}
D_1^* &=\max_{P}D_1(P,U)\\
&=-\frac{1}{2}\left[\frac{d\,!+1}{d\,!}\log(d\,!+1) - \log d\,! - 2 \log 2\right]\,,
\end{split}
\end{equation}
is a normalization constant (obtained by calculating $D_1(P,U)$ when one component of $P$ is one and all others are zero). In spite of the fact that the statistical complexity $C_1(P)$ is defined by the product of $D_1(P,U)$ and $H_1(P)$, $C_1(P)$ is not a trivial function of $H_1(P)$ in the sense that, for a given value of $H_1(P)$, there is a range of possible values for $C_1(P)$~\cite{MartinPlastinoRosso2006}. Because of that Rosso~\textit{et al.}~\cite{RossoLarrondoMartinPlastinoFuentes2007} proposed a representation space composed of the values of $C_1(P)$ versus $H_1(P)$, building up the complexity-entropy causality plane and finding that chaotic and stochastic time series occupy different regions of this diagram.

Despite being successfully applied for studying several systems, the values of $C_1(P)$ and $H_1(P)$ are not enough for capturing different scales of the system dynamics as well as different meanings for complexity. Because of that, we propose to replace the normalized Shannon entropy (Eq.~\ref{eq:shannon}) and the statistical complexity (Eq.~\ref{eq:statcomplexity}) by mono-parametric generalizations based on the Tsallis $q$-entropy. This entropic form is a generalization of the Shannon entropy and can be defined as~\cite{Tsallis1988,Tsallis}
\begin{equation}\label{eq:tsallisentropy}
S_q(P)=\sum_{j=1}^{d!} p_j\log_q\frac{1}{p_j}\,,
\end{equation}
where $q$ is a real parameter and $\log_qx=\int_1^xt^{-q}\,dt$ is the \emph{$q$-logarithm} ($\log_qx=\frac{x^{1-q}-1}{1-q}$ if $q\ne 1$ and $\log_1x=\log x$ for any $x>0$)~\cite{Tsallis}. We will use the convention $0\log_q(1/0)=0$ whenever $q>0$. It is worth noting that $S_1$ is the Shannon entropy.

Once defined the $q$-entropy, we further consider its normalized version (analogously to the Eq.~\ref{eq:shannon}):
\begin{equation}\label{eq:tsallisentropynormal}
H_q(P)=\frac{S_q(P)}{S_q(U)}\,,
\end{equation}
where $S_q(U)=\log_q d!$ is the maximum value of the $q$-entropy~\cite{Tsallis}. Furthermore, by following the developments of Martin, Plastino and Rosso~\cite{MartinPlastinoRosso2006}, we assume the generalized version of the statistical complexity (Eq.~\ref{eq:statcomplexity}) -- the $q$-complexity -- to be
\begin{equation}\label{eq:qstatcomplexity}
C_q(P)=\frac{D_q(P,U)H_q(P)}{D_q^*}\,,
\end{equation}
where
\begin{equation}
\begin{split}
D_q(P,U)&=\frac{1}{2}K_q\!\left(\!P\biggl|\frac{P+U}{2}\right) + \frac{1}{2}K_q\!\left(\!U\biggl|\frac{P+U}{2}\right)\\
&=-\frac{1}{2}\sum_{\substack{i=1\\p_i\ne 0}}^{d!}p_i\log_q\frac{p_i+1/d!}{2p_i}\\
&\quad-\frac{1}{2}\sum_{i=1}^{d!}\frac{1}{d!}\log_q\frac{p_i+1/d!}{2/d!}
\label{Dq}
\end{split}
\end{equation}
is a distance between $P$ and $U$, {$K_q(P|R)=-\sum p_i \log_q r_i/p_i$} a generalization of Kullback-Leibler divergence within the Tsallis formalism~\cite{MartinPlastinoRosso2006}, and
\begin{equation}
\begin{split}
D_q^*&=\max_{P}D_q(P,U)\\
&=\frac{2^{2-q}d!-(1+d!)^{1-q}-d!(1+1/d!)^{1-q}-d!+1}{(1-q)2^{2-q}d!}
\end{split}
\label{Dq*}
\end{equation}
is a normalization constant. 

Thus, the quantities $H_q$ and $C_q$ are generalizations of the normalized entropy and complexity within the symbolic approach of Bandt and Pompe~\cite{BandtPompe2002}, introduced by Rosso~\textit{et al.}~\cite{RossoLarrondoMartinPlastinoFuentes2007}, which are included as the particular case~$q=1$. Here, we are interested in the parametric representation of the ordered pairs $(H_q(P),C_q(P))$ on $q>0$ for a fixed distribution $P$. We call this curve the \emph{$q$-complexity-entropy curve}, and we shall see that this representation has superior capabilities of distinguishing time series when compared with the point $(H_1(P),C_1(P))$ in the complexity-entropy causality plane.

Before we proceed to the applications, let us enumerate some general properties of the $q$-complexity-entropy curves. For a given probability distribution $P=\{p_j(\pi_j)\}_{j=1,\ldots,d!}$, let $r$ be the number of non-zero components of $P$ (that is, the number of permutations $\pi_j$ that actually occurs in the time series) and $\gamma=\frac{r-1}{d!-1}$ (essentially the fraction of occurring permutations among all $d!$ possible). From the definitions of $H_q$ and $C_q$, it is not difficult to prove the following statements (see Appendix~\ref{app:Limiting_expression}):
\begin{enumerate}
\item If $r=1$ then $H_q(P)=0$ and $C_q(P)=0$ for any $q>0$;
\item $H_q(P)\to \gamma$ and $C_q(P)\to \gamma(1-\gamma)$ as $q\to 0^+$;
\item If $r>1$ then $H_q(P)\to 1$ and $C_q(P)\to 1-\gamma$ as $q\to\infty$.
\end{enumerate}
These general properties of $H_q$ and $C_q$ have the following consequences for the $q$-complexity-entropy curves:
\begin{enumerate}
\item The $q$-complexity-entropy curve of a time series that only displays one permutation $\pi_j$ (that is, for $r=1$) collapses onto the point~$(0,0)$;
\item For a time series that has all possible permutations $\pi_j$ (that is, $r=d!$ and $\gamma=1$), the $q$-complexity-entropy curve is a loop that starts at the point $(1,0)$ for $q=0^+$ and ends at the same point for $q\to\infty$;
\item For a time series that does not display all permutations $\pi_j$, the $q$-complexity-entropy curve starts at the point $(\gamma,\gamma(1-\gamma))$ for $q=0^+$ and ends at the point $(1,1-\gamma)$ for $q\to\infty$. Here $0<\gamma<1$, and the number of occurring permutations $r$ can be obtained from $\gamma$ via $r=(d!-1)\gamma+1$.
\end{enumerate}

We shall see that noisy time series are usually characterized by closed $q$-complexity-entropy curves, whereas chaotic time series have open curves (especially for large embedding dimensions). This last feature is related to the existence of forbidden ordinal patterns in the chaotic dynamics that is common in several chaotic maps~\cite{Amigo2006,Amigo2007,Amigo2008,Amigo}, but that can also appear in stochastic processes depending on the time series length~\cite{Rosso2012,Rosso201242,Carpi20102020}.

\section{Numerical applications}
In this section, we present several applications of the $q$-complexity-entropy curve for numerically-generated time series of stochastic and chaotic nature. 

\subsection{Fractional Brownian motion}
As a first application, we study time series generated from the fractional Brownian motion~\cite{Mandelbrot}. The fractional Brownian motion is a stochastic process that has stationary, long-range correlated, and Gaussian increments. It is usually defined in terms of a parameter $h$ (the so-called Hurst exponent): for $h<1/2$, the fractional Brownian motion is anti-persistent, meaning that positive increments are followed by negative increments (or vice versa) more frequently than by chance; while for $h>1/2$, it is persistent, meaning that positive increments are followed by positive increments and negative increments are followed by negative increments more frequently than by chance. Also, we have fully persistent motion in the limit of $h\to1$, whereas the usual Brownian motion is recovered in the limit of $h\to1/2$.

\begin{figure}[!ht]
\centering
\includegraphics[scale=0.35]{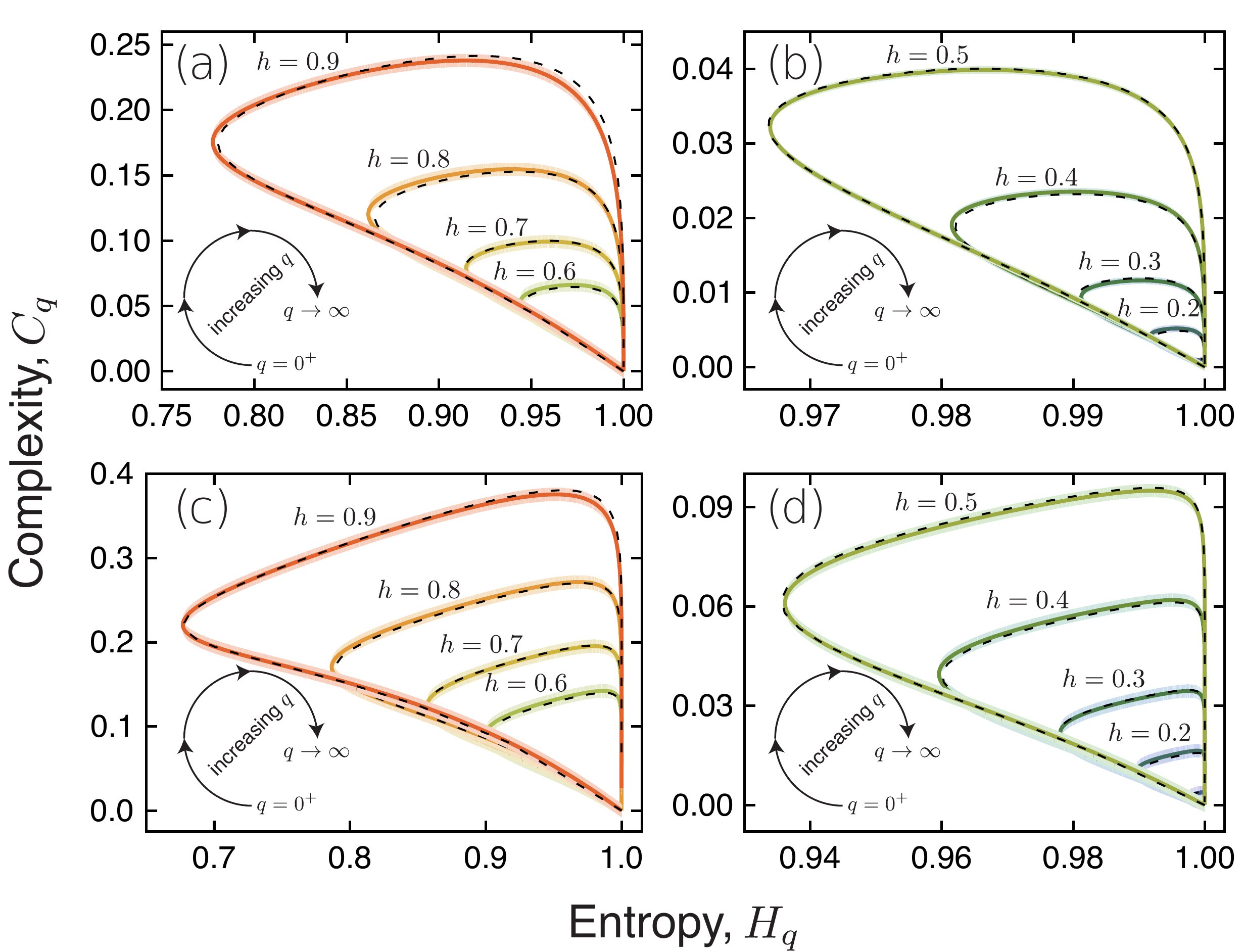}
\caption{
{Dependence of the entropy $H_q$ and complexity $C_q$ on the parameter $q$ for the fractional Brownian motion.} Panels (a) and (b) show the $q$-complexity-entropy curves with $d=3$ and several values of the Hurst exponent $h$ ($h$ in $(0.2, 0.3, 0.4, 0.5, 0.6, 0.7, 0.8, 0.9)$, as indicated in the plots). The values of $q$ are increasing (from $q=0^+$ to $q=1000$ in size steps of $10^{-4}$) in the clockwise direction. Panels (c) and (d) show the same, but with embedding dimension $d=4$. All colored curves represent the average value of $H_q$ and $C_q$ over one hundred realizations of a fractional Brownian walker with $2^{17}$ steps. The shaded areas indicate the 95\% confidence intervals estimated via bootstrapping (over one hundred independent realizations). The dashed lines are the exact results calculated by using the probabilities of Bandt and Shiha~\cite{BandtShiha2007} (see also Appendix~\ref{app:ordinal_probabilities}). It is worth noting that all curves form loops in this representation space, starting at $(1,0)$ for $q=0^+$ and ending at $(1,0)$ for $q\to\infty$.
}
\label{fig:1}
\end{figure}

In order to calculate the $q$-complexity-entropy curves associated with the fractional Brownian motion, we numerically generate time series of length $2^{17}$ following the procedure of Hosking~\cite{Hosking} for different values of the Hurst exponent $h$. Figures~\ref{fig:1}(a) and~\ref{fig:1}(b) show these curves for the embedding dimension $d=3$ and $h$ in $(0.2,0.3,\dots,0.9)$, while Figs.~\ref{fig:1}(c) and~\ref{fig:1}(d) are the same for $d=4$. These plots show the average values (over one hundred realizations) of the ordered pairs $(H_q(P),C_q(P))$, with $q$ from $10^{-4}$ (assumed to be $q=0^+$) to $1000$ in steps of $10^{-4}$. We note that all $q$-complexity-entropy curves are closed, indicating that time series of length $2^{17}$ of the fractional Brownian motion displays all possible permutations $\pi_j$ for $d=3$ and $d=4$. The presence of forbidden ordinal patterns in the fractional Brownian motion was studied by Rosso~\textit{et al.}~\cite{Rosso2012,Rosso201242} and Carpi~\textit{et al.}~\cite{Carpi20102020}, where they observed that the number of forbidden ordinal patterns decreases with the time series length with a rate that depends on the Hurst exponent $h$. In particular, Carpi~\textit{et al.}~\cite{Carpi20102020} showed that for $d=4$ and very small time series (around one hundred steps), the fractional Brownian motion may have a few number of forbidden patterns; however, this number vanishes for series of length larger than 500 terms, which agrees with our findings. Also, for the fractional Brownian motion is possible to obtain the exact expression for the $q$-complexity-entropy curves, because Bandt and Shiha~\cite{BandtShiha2007} have calculated the exact form of the probability distribution $P=\{p_j(\pi_j)\}_{j=1,\ldots,d!}$ for $d=3$ and $d=4$ (all values of $p_j(\pi_j)$ are provided in the Appendix~\ref{app:ordinal_probabilities}). By using these distributions, the expressions~(\ref{eq:tsallisentropynormal}) and~(\ref{eq:qstatcomplexity}) lead to the exact values for the ordered pairs $(H_q(P),C_q(P))$ for all $q$. The dashed lines in Fig.~\ref{fig:1} show the exact $q$-complexity-entropy curves for the fractional Brownian motion, where we observe an excellent agreement with the numerical results.

The results of Fig.~\ref{fig:1} also reveal that the $q$-complexity-entropy curves distinguish the different values of the Hurst exponent $h$. We observe that the larger the value of $h$, the broader the loop formed by the $q$-complexity-entropy curve. We also find that the normalized entropy $H_q$ as a function of $q$ has a minimum value at $q=q_H^*$ and that the complexity $C_q$ as a function of $q$ has a maximum value at $q=q_C^*$, both extreme values of $q$ depend on the Hurst exponent $h$ and also on the embedding dimension $d$. This dependence is shown in Fig.~\ref{fig:2} for $d=3$ and $d=4$, where a good agreement between the exact and the numerical values of these extreme values is observed. We further notice that $q_H^*$ increases with $h$ up to a maximum and then starts to decrease [Figs.~\ref{fig:2}(a)~and~\ref{fig:2}(c)]; whereas $q_C^*$ is a monotonically increasing function of the Hurst exponent $h$ [Figs.~\ref{fig:2}(b)~and~\ref{fig:2}(d)]. 

\begin{figure}[!ht]
\centering
\includegraphics[scale=0.33]{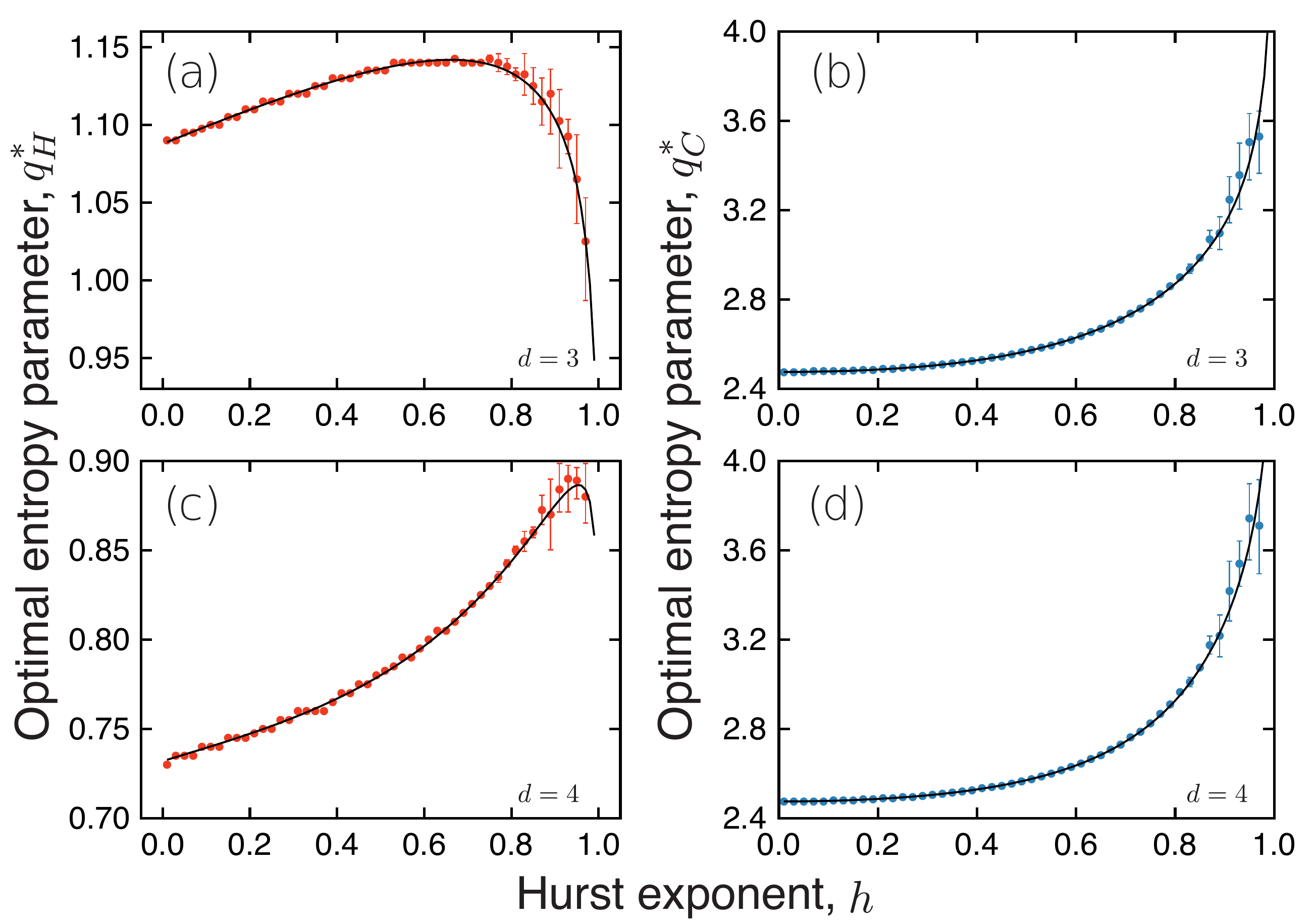}
\caption{
{Comparison between the extreme values $q^*_H$ and $q^*_C$ obtained from the simulations and the exact results for the fractional Brownian motion.} Panel (a) shows the values of $q=q^*_H$ for which $H_q$ reaches a minimum as a function of the Hurst exponent $h$ and $d=3$. Each dot corresponds to the average value of $q^*_H$ obtained from one hundred realizations of a fractional Brownian walker with $2^{17}$ steps. Error bars stand for 95\% confidence intervals estimated via bootstrapping (over one hundred independent realizations). Panel (b) shows the values of $q=q^*_C$ for which $C_q$ reaches a maximum as a function of the Hurst exponent $h$ and $d=3$. Again, the dots are the average value calculated from one hundred independent realizations of a fractional Brownian walker with $2^{17}$ steps, and the error bars are 95\% confidence intervals. In both plots, the continuous lines are the exact results. Panels (c) and (d) are the analogous of (a) and (b) when considering the embedding dimension $d=4$. In all cases, we note an excellent agreement between the simulations and the exact results.
}
\label{fig:2}
\end{figure}

The extreme values of the normalized $q$-entropy and the $q$-complexity-entropy ($H_{q_H^*}$ and $C_{q_C^*}$) represent the largest contrast between $H_{q}$ (as well as $C_{q}$) calculated for the system distribution $P=\{p_j(\pi_j)\}_{j=1,\ldots,d!}$ and the uniform distribution $U=\{1/d\,!\}_{j=1, \dots d!}$. In the context of ecological diversity, the values of $H_{q_H^*}$ were found to enhance contrast among ecological communities when compared with usual diversity indexes~\cite{Mendes_etal2008}. Similarly, the values of $H_{q_H^*}$ and $C_{q_C^*}$ may enhance the differentiation among time series of the fractional Brownian motion with different Hurst exponents. In order to verify this hypothesis, we test the performance of the values $H_{q_H^*}$ and $C_{q_C^*}$ (in comparison with $H_{1}$ and $C_{1}$) for classifying time series of the fractional Brownian motion with different Hurst exponents. For this, we generate an ensemble with one hundred time series for each value of the Hurst exponent $h$ in $(0.03, 0.05,\dots, 0.97)$. Next, we train a $k$-nearest neighbors algorithm~\cite{Pang} in a $3$-fold cross-validation strategy, considering these $48$ different values of $h$ as possible classes for the algorithm. This machine learning classifier is one of the simplest algorithms for supervised learning~\cite{Pang}; it basically assigns a class (here the value of the Hurst exponent) to an unlabeled point based on the class of the majority of the nearby points. 

\begin{figure}[!ht]
\centering
\includegraphics[scale=0.23]{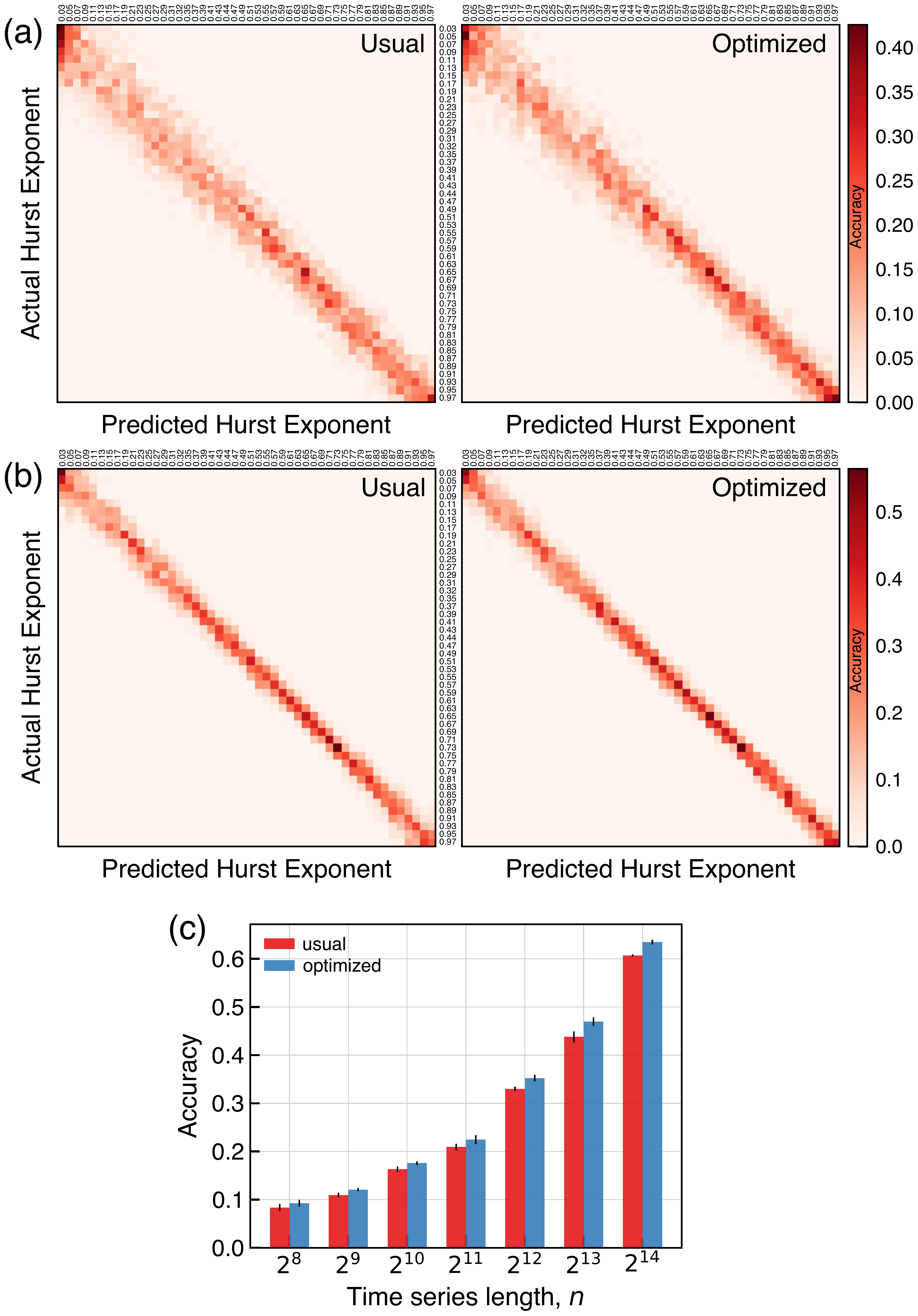}
\caption{
{Predicting the Hurst exponent based on the values of $H_q$ and $C_q$ via the nearest neighbors algorithm.} (a) Confusion matrices obtained from the nearest neighbors algorithm when considering the values of $H_q$ and $C_q$ for $q=1$ (usual case -- left panel) and the optimized values of $H_{q^*_H}$ and $C_{q^*_C}$ (optimized case -- right panel) for time series of length $n=2^{10}$. The rows of these matrices represent the actual Hurst exponents (the value employed in the simulation) and the columns represent the values predicted by the machine learning algorithm. Values of the Hurst exponents varies from $0.03$ to $0.97$ in steps of size $0.02$. The color code indicates the fraction of occurrences for each combination of actual and predicted Hurst exponent. In (b) we show the same analysis for longer times series ($n=2^{12}$). We note that in both classification scenarios the predicted values are always very close to the actual values (that is, the confusion matrices have ``diagonal stripes''). However, we further notice that the ``diagonal stripe'' is narrower in the optimized case (especially for persistent processes). (c) Overall accuracy (fraction of correctly classified Hurst exponents) for the two classification scenarios (error bars are the standard-errors) for different time series length. We notice that the use of $H_q$ and $C_q$ with $q=q^*_H$ and $q=q^*_C$ (optimized case) provides a greater overall accuracy when compared with the case $q=1$, regardless of the length $n$.
}
\label{fig:3}
\end{figure}

Figure~\ref{fig:3}(a) shows the confusion matrices (that is, the fractions of time series with a particular Hurst exponent that are classified with a given Hurst exponent -- accuracy) when considering $H_{1}$ and $C_{1}$ (usual causality plane, $q=1$ -- left panel) and the $H_{q_H^*}$ and $C_{q_C^*}$ (optimized causality plane -- right panel) for time series with $n=2^{10}$ terms. We observe that both matrices have non-zero elements only around the diagonal, indicating that the misclassified Hurst exponents are labeled with values close to actual values. We also note that the ``diagonal stripe'' of these matrices is narrow for the optimized causality plane (specially for $h>0.5$), showing that the optimized values enhance the performance of the classification task. Figure~\ref{fig:3}(b) shows the same analysis with longer time series ($n=2^{12}$), where we note the narrowing of the ``diagonal stripe'' of the confusion matrices. This happens because the variance in the estimated values of $H_{q}$ and $C_{q}$ decreases with the length $n$ of the time series, enhancing the performance of the classifiers in both scenarios. However, we still observe that the accuracy is larger when employing the optimized values $H_{q_H^*}$ and $C_{q_C^*}$. In fact, Fig.~\ref{fig:3}(c) shows that overall accuracy is always enhanced (regardless of $n$) when considering the optimized causality plane in comparison with usual causality plane.

\subsection{Harmonic noise}
For another numerical application, we consider times series generated from the harmonic noise~\cite{Geier1990}. This stochastic process is a generalization of the Ornstein-Uhlenbeck process~\cite{gardiner_handbook} and can be defined by the following system of Langevin equations~\cite{Geier1990}
\begin{equation}\label{eq:langevin_harm}
\begin{split}
\frac{dy}{dt} &= s\\
\frac{ds}{dt} &= - \Gamma s - \Omega^2 y + \sqrt{2 \varepsilon}\, \Omega^2 \xi(t)\,,
\end{split}
\end{equation}
where $\xi(t)$ is a Gaussian noise with zero mean, $\langle\xi(t)\rangle=0$, (here $\langle\dots\rangle$ stands for ensemble average) and uncorrelated, $\langle\xi(t)\xi(t')\rangle=\delta(t-t')$. That is, a harmonic oscillator driven by a white noise. This noise has an oscillating correlation function given by~\cite{Geier1990}
\begin{equation}\label{eq:harm_cor}
\begin{split}
\langle y(t) y(t+\tau)\rangle &= \frac{\varepsilon \Omega^2}{\Gamma} \exp\left(-\frac{\Gamma}{2}\tau\right)\\
&\times \left[\cos(\omega \tau) + \frac{\Gamma}{2\omega}\sin(\omega \tau)\right]\,,
\end{split}
\end{equation}
where 
\begin{equation}\label{eq:harm_omega_rel}
\omega = \sqrt{\Omega^2 - (\Gamma/2)^2}
\end{equation}
is the frequency of oscillation. In practical terms, this noise is a mixture of random and periodic behaviors. Notice that the Ornstein-Uhlenbeck process is recovered in the limit $\Omega\to\infty$ and $\Gamma\to\infty$, while the ratio $\Gamma/\Omega^2$ remains fixed (for this case, $\langle y(t) y(t+\tau)\rangle\sim \exp(-\tau)$). 

In order to produce time series from the harmonic noise, we integrate the system of equations~(\ref{eq:langevin_harm}) by using the Euler method with a step size $dt=10^{-3}$ (an approach that produces a good agreement between the exact correlation function of Eq.~\ref{eq:harm_cor} and the numerical results) up to maximum integration time of $1320$ and for particular values of parameters $\Gamma$, $\Omega$, and $\varepsilon$. In particular, we first investigate the role of the frequency $\omega$ on the form of the $q$-complexity-entropy curve. To do so, we choose $\varepsilon=1$, $\Gamma=0.05$, and several values of $\omega$ ranging from $1$ to $60$ ($\Omega$ is obtained from Eq.~\ref{eq:harm_omega_rel}). For these parameters, the shape of the correlation function is similar to an underdamped simple harmonic motion. Figure~\ref{fig:4} shows some $q$-complexity-entropy curves for the embedding dimension $d=3$. We note that all curves form loops and the broader the loop, the smaller the value of $\omega$. Also, the largest contrasts between the values of $\omega$ occur around the regions of minimum entropy and maximum complexity (as indicated by the insets). The curves are strongly overlapped for very small or very large values of $q$. We further observe that values of the complexity for $q=1$ (black dots in the first inset) practically do not change with $\omega$. 

\begin{figure}[!ht]
\centering
\includegraphics[scale=0.33]{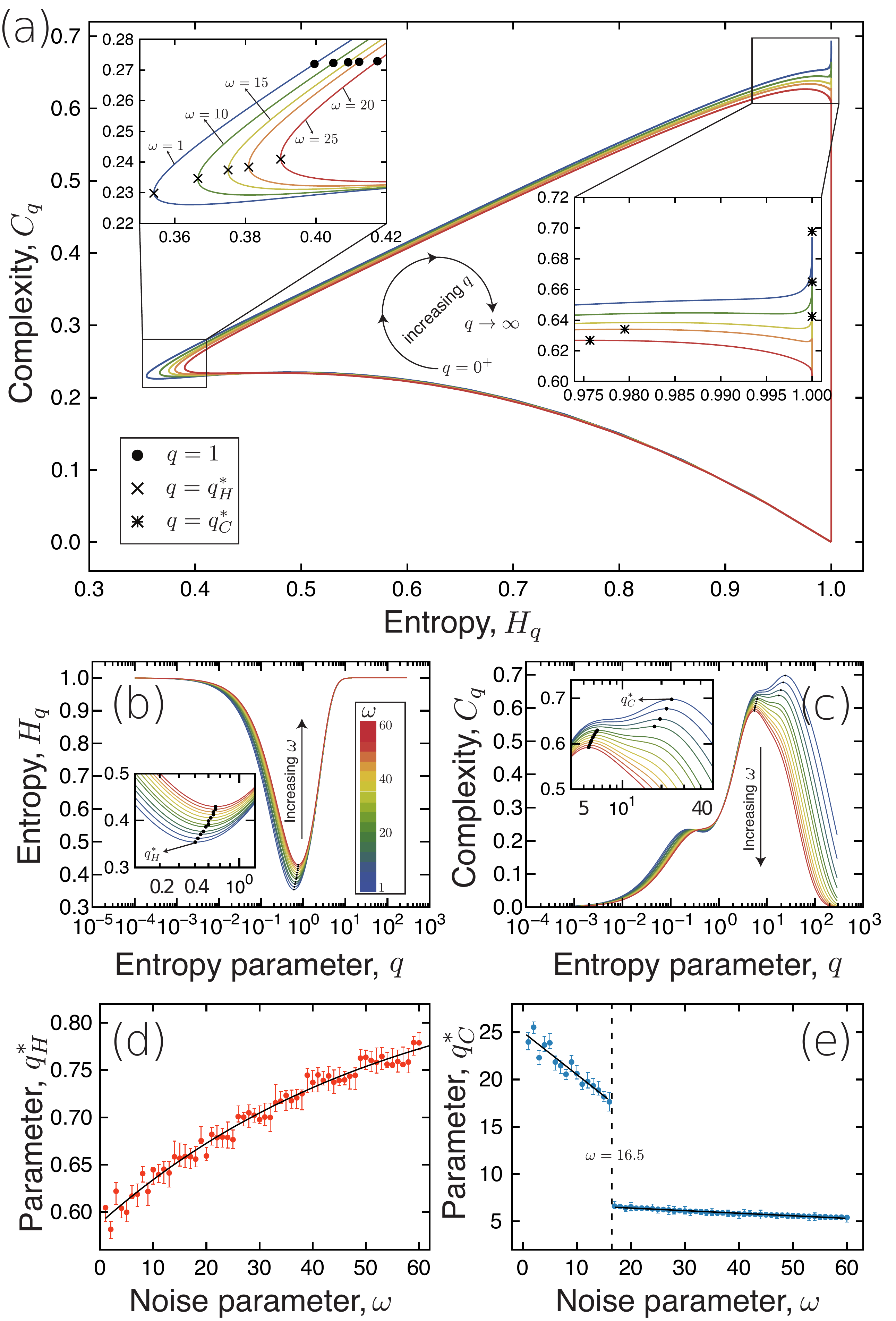}
\caption{
{Dependence of the entropy $H_q$ and complexity $C_q$ on the parameter $q$ for the harmonic noise: changes with the frequency parameter $\omega$.} Panel (a) shows the $q$-complexity-entropy curves with embedding dimension $d=3$, $\Gamma=0.05$, $\varepsilon=1$, and some values of the parameter $\omega$ (shown in the plot). The values of $q$ are increasing (from $q=0^+$ to $q=1000$ in size steps of $10^{-4}$) in the clockwise direction. The two insets highlight the regions of the causality plane where the entropy reaches a minimum ($q=q^*_H$, indicated by cross markers) and the complexity passes to its maximum value ($q=q^*_C$, indicated by asterisk markers). The points ($H_q$, $C_q$) for $q=1$ are indicated by black dots. Panel (b) shows the dependence of $H_q$ on $q$ for several values of $\omega$ (indicated by the color code) and the inset highlights the region where the minimums occur ($q=q^*_H$, indicated by black dots). Panel (c) shows the dependence of $C_q$ on $q$ for several values of $\omega$ (the same of panel b) and the inset highlights the region where the maximums occur ($q=q^*_C$, indicated by black dots). Panels (d) and (e) show the dependence of the extreme values of $q$ ($q^*_H$ and $q^*_C$) on the frequency parameter $\omega$ (the markers are the average values over one hundred realizations of the harmonic noise with maximum integration time of 1320 and step size of $10^{-3}$, and the error bars stand for 95\% bootstrap confidence intervals). We note that $q^*_H$ monotonically increases with $\omega$ in a relationship that is approximated by an exponential approach to the value $q^*_H=0.867$ (that is, $q^*_H=0.867-0.279 e^{-0.018 \omega}$, as indicated by the continuous line). We further notice that $q^*_C$ decreases with $\omega$ and that around $\omega=16.5$ there is a discontinuous behavior. The continuous lines in this last plot are linear approximations to the behavior of $q^*_C$.
}
\label{fig:4}
\end{figure}
\clearpage
Figures~\ref{fig:4}(b)~and~\ref{fig:4}(c) depict the individual behavior of $H_q$ and $C_q$ versus $q$ (now for more values of $\omega$), where the insets show the form of these curves around their extreme values (that are indicated by small dots). Finally, in Figs.~\ref{fig:4}(d)~and~\ref{fig:4}(e) we study the dependence of the extreme values $q_H^*$ and $q_C^*$ on the parameter $\omega$. We find that $q_H^*$ monotonically increases with $\omega$ in a relationship that can be approximated by an exponential approach to the value $q_H^*=0.867$. The shape of $q_C^*$ is more intriguing because it suddenly changes around the value $\omega\approx16.5$, a behavior that is similar to a phase transition in a bistable system.

We further investigate the shape of the $q$-complexity-entropy curves in a situation that is closer to a pure Ornstein-Uhlenbeck process, that is, a process with a correlation function that decays exponentially. For this, we fix $\omega=10^{-4}$ and choose different values for $\Gamma$, ranging from close to zero up to $10^4$ (again $\varepsilon=1$ and $\Omega$ is obtained from Eq.~\ref{eq:harm_omega_rel}). The small value of $\omega$ ensures that the oscillation period of the correlation function (Eq.~\ref{eq:harm_cor}) is much larger than the integration time. Figure~\ref{fig:5}(a) shows some $q$-complexity-entropy curves for the embedding dimension $d=3$ and different values of $\Gamma$ ranging from $0.55$ to $100$. These curves form loops whose broadness decreases as $\Gamma$ increases; in fact, the form of these curves is approaching a limit loop that is similar to the one observed for the fractional Brownian motion with $h=1/2$. Thus, the $q$-complexity-entropy curve can also distinguish among different time series with short-range correlations. Figures~\ref{fig:5}(b)~and~\ref{fig:5}(c) show the individual behavior of $H_q$ and $C_q$ versus $q$, where we find that the extreme values $q_H^*$ and $q_C^*$ depend on $\Gamma$, as illustrated in Figs.~\ref{fig:5}(d)~and~\ref{fig:5}(e). For small values of $\Gamma$, $q_H^*$ logarithmically increases with $\Gamma$ up to $\Gamma\approx300$, where it saturates around $q_H^*\approx1.12$. Similarly, $q_C^*$ logarithmically decreases with $\Gamma$ up to $\Gamma\approx1000$, where it saturates around $q_C^*\approx2.53$. These limit values for $q_H^*$ and $q_C^*$ are very close to those obtained for the fractional Brownian motion with $h=0.5$ [a random walk, see Figs.~\ref{fig:2}(a)~and~\ref{fig:2}(b)].

\begin{figure}[!ht]
\centering
\includegraphics[scale=0.33]{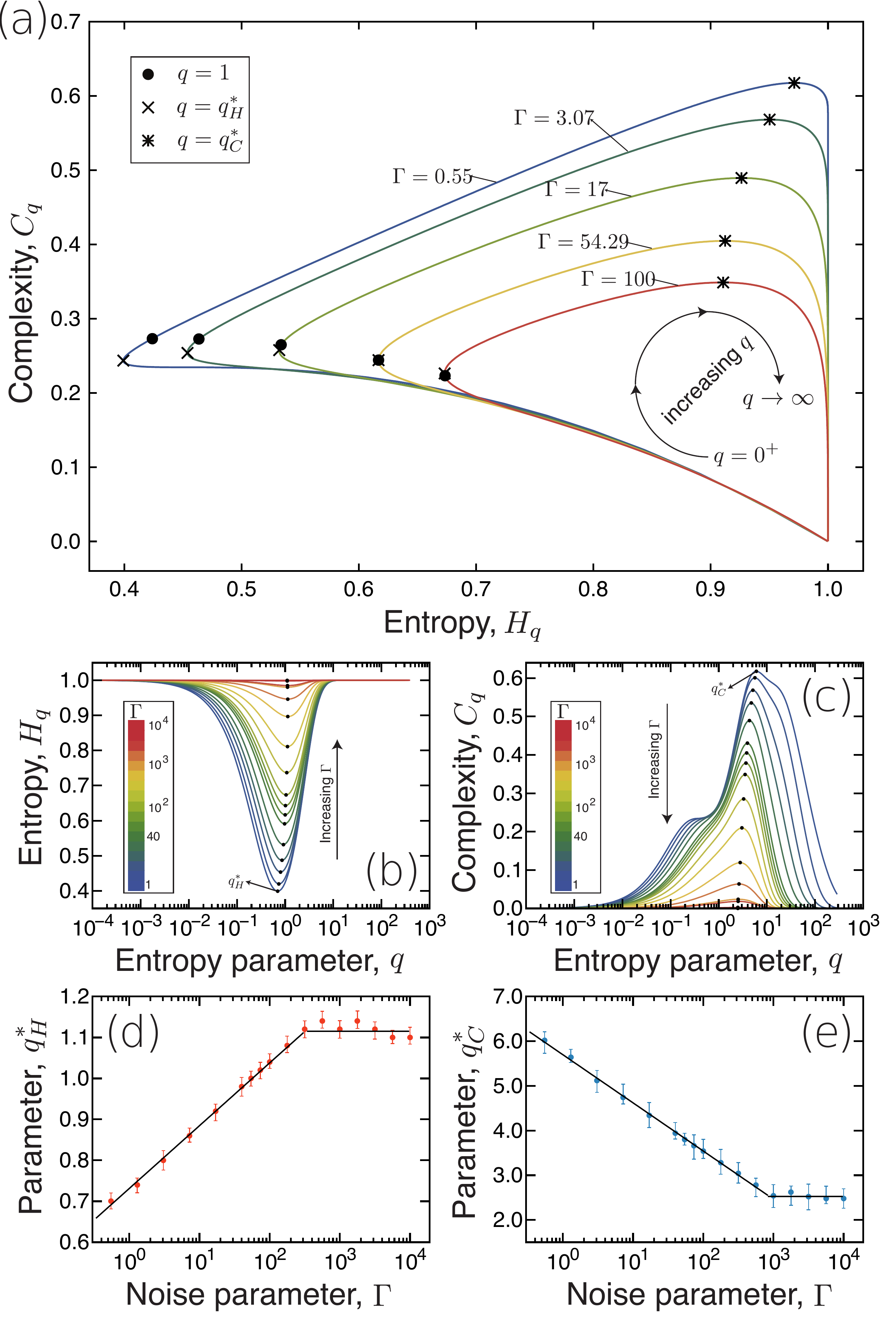}
\caption{
{Dependence of the entropy $H_q$ and complexity $C_q$ on the parameter $q$ for the harmonic noise: changes with the damping coefficient $\Gamma$.} Panel (a) shows $q$-complexity-entropy curves with embedding dimension $d=3$, $\omega=10^{-4}$, $\varepsilon=1$, and for some values of the parameter $\Gamma$ (shown in the plot). The values of $q$ are increasing (from $q=0^+$ to $q=1000$ in size steps of $10^{-4}$) in the clockwise direction. The black dots indicate the points ($H_q$, $C_q$) for $q=1$, cross markers for $q=q^*_H$, and asterisk markers for $q=q^*_C$. Panel (b) shows the dependence of $H_q$ on $q$ for several values of $\Gamma$ (indicated by the color code), where the black dots indicate the values of $q=q^*_H$ that minimize $H_q$. Panel (c) shows the dependence of $C_q$ on $q$ for several values of $\Gamma$ (the same of panel b), where the black dots indicate the values of $q=q^*_C$ that maximize $C_q$. Panels (d) and (e) show the dependence of the extreme values of $q$ ($q^*_H$ and $q^*_C$) on the damping parameter $\Gamma$ (the markers are the average values over one hundred realizations of the harmonic noise with maximum integration time of 1320 and step size of $10^{-3}$, and the error bars stand for 95\% bootstrap confidence intervals). We note that $q^*_H$ logarithmically increases with $\Gamma$ up to $\Gamma\approx300$ ($q^*_H = 0.73+0.07*\ln(\Gamma)$, as indicated by the continuous line), where it saturates around $q^*_H\approx1.12$ (continuous line). We further notice that $q^*_C$ logarithmically decreases with $\Gamma$ up to $\Gamma\approx1000$ ($q^*_C = 5.70-0.47*\ln(\Gamma)$, as indicated by the continuous line), where it saturates around $q^*_C\approx2.53$ (continuous line).
}
\label{fig:5}
\end{figure}

\subsection{Chaotic maps at fully developed chaos}
\begin{figure*}[!ht]
\centering
\includegraphics[scale=0.35]{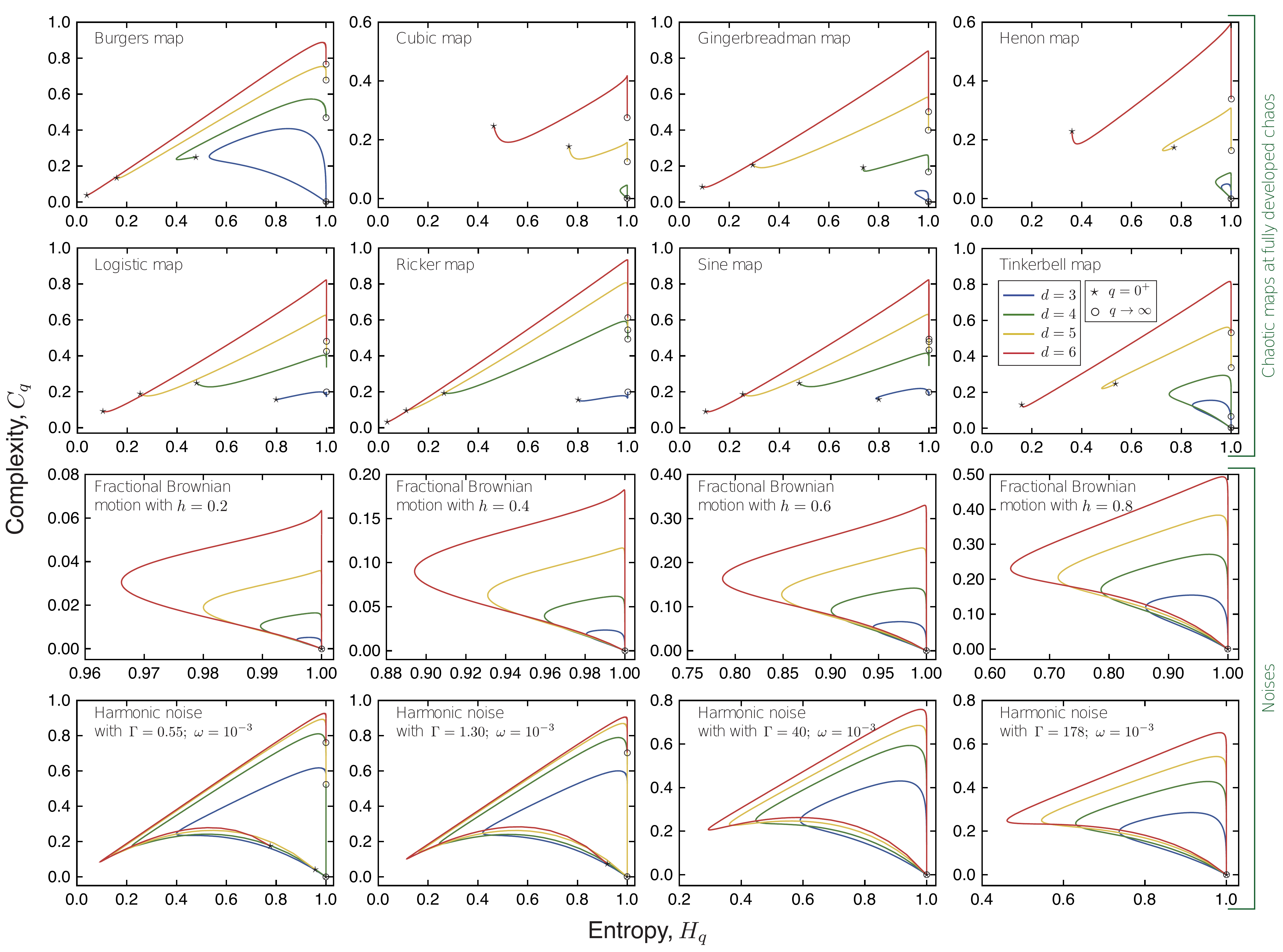}
\caption{
{Dependence of the entropy $H_q$ and complexity $C_q$ on the parameter $q$ for chaotic maps at fully developed chaos and stochastic processes.} Each plot shows the $q$-complexity-entropy curve for a chaotic map (first two rows) or a stochastic process (last two rows) for embedding dimensions $d=3,4,5,\text{and}~6$ (the different colors). The first two rows show the results for the chaotic maps: Burgers, cubic, Gingerbreadman, Henon, logistic, Ricker, sine and Tinkerbell at fully developed chaos (see Appendix~\ref{app:chaotic_maps} for more detail). The third row refers to the fractional Brownian motion (Hurst exponent $h$ is indicated in the plots) and the last row refers to the harmonic noise (parameters $\Gamma$ and $\omega$ are indicated in the plots). In each panel, the star markers indicate the points $(H_q,C_q)$ for $q=0^+$, while the open circles are the same for $q\to\infty$. We note that stochastic processes are mostly characterized by loops for all embedding dimensions, whereas chaotic maps usually form open curves in the causality plane. We further note that, differently from stochastic processes, $H_q$ of chaotic maps does not exhibit a minimum value for all embedding dimensions. 
}
\label{fig:6}
\end{figure*}

We now focus on analyzing the shape of the $q$-complexity-entropy curves for time series associated with chaotic processes. To do so, we generate time series by iterating eight chaotic maps: Burgers, cubic, Gingerbreadman, Henon, logistic, Ricker, sine, and Tinkerbell. We set their parameters to ensure a fully chaotic regime and iterate over $2^{17}+10^{4}$ steps, dropping the initial $10^{4}$ steps for avoiding transient behaviors. Also, for the two-dimensional maps (Burgers, Gingerbreadman, Henon and Tinkerbell), we have considered the squared sum of the two coordinates. The definition of each map and the parameters employed are given in the Appendix~\ref{app:chaotic_maps}. Figure~\ref{fig:6} shows the $q$-complexity-entropy curves for each map and for the embedding dimensions between $d=3$ and $d=6$. Differently from our previous results for noises, these curves do not form loops for all embedding dimensions, showing that there are permutations of $\pi_j$ that never appear in these time series. 

For comparison, we also show in Fig.~\ref{fig:6} some $q$-complexity-entropy curves for the fractional Brownian motion and harmonic noise, calculated from time series of the same length used for the chaotic maps. For the fractional Brownian motion, loops are observed for all values of $d$ and $h$; however, for the harmonic noise there are some open $q$-complexity-entropy curves (when $\omega=10^{-3}$ and $\Gamma=0.55$ or $\Gamma=1.30$), indicating that this noise presents forbidden permutations, even for time series of length $2^{17}$. As reported by Rosso~\textit{et al.}~\cite{Rosso2012,Rosso201242} and Carpi~\textit{et al.}~\cite{Carpi20102020} for the fractional Brownian motion, we expect the number of forbidden permutations to vanish with the length of the time series, and loops should appear for longer time series. For instance, we find that the curves for the harmonic noise shown in Fig.~\ref{fig:6} become loops for time series a thousand times longer, which does not happen for the chaotic maps. Thus, the shape of the $q$-complexity-entropy curve (closed or open) can be used as an indicative of chaos (open curves) or stochasticity (closed curves). Another characteristic that can distinguish between chaotic and stochastic time series is the existence of a minimum value for the normalized entropy $H_q$. We note that a minimum value exists in all time series from harmonic and fractional noise for $3\leq d\leq6$, which does not happen for the chaotic maps ($H_q\to1$ monotonically for most $d$ values).

\subsection{Logistic map}
Still on chaotic processes, we investigate the logistic map in more detail. This map is a quadratic recurrence equation defined as~\cite{may1976simple}
\begin{equation}
y_{k+1} = a\,y_k (1-y_k)\,,
\end{equation}
where $a$ is a parameter whose values of interest are in the interval $0\leq a\leq4$ (for which $0\leq y_k\leq1$). Depending on $a$, this map can exhibit simple periodic behavior (\textit{e.g.} $a=3.05$), stable cycles of period $m$ (\textit{e.g.} $m=4$ for $a=3.5$ and $m=8$ for $a=3.55$), and chaos (most values of $a>3.56994567\dots$ and $a=4$). 

This map is particularly interesting for our study because we can find the exact expression of the $q$-complexity-entropy curve when $d=3$ and $a=4$. Amig\'o~\textit{et al.}~\cite{Amigo2006,Amigo2007,Amigo2008,Amigo} have shown that the list $\{y_k,y_{k+1},y_{k+2}\}$ always corresponds to the ordinal pattern $\{0,1,2\}$ when $0<y_k<\frac{1}{4}$. Similarly, the ordinal pattern
$\{0,2,1\}$ occurs for $\frac{1}{4}<y_k<\frac{5-\sqrt{5}}{8}$,
$\{2,0,1\}$ for $\frac{5-\sqrt{5}}{8}<y_k<\frac{3}{4}$,
$\{1,0,2\}$ for $\frac{3}{4}<y_k<\frac{5+\sqrt{5}}{8}$,
$\{1,2,0\}$ for $\frac{5+\sqrt{5}}{8}<y_k<1$, 
and the ordinal pattern $\{2,1,0\}$ never appears. 
Combining these results with the fact that the probability distribution of $y_k$ is a beta distribution~\cite{jakobson1981absolutely}, $\rho(y)=[\pi \sqrt{y(1-y)}]^{-1}$, we can find the probability distribution $P=\{p_j(\pi_j)\}_{j=1,\ldots,d!}$ by integrating the beta distribution over each one of the previous intervals of $y_k$ (for instance, the probability associated with the pattern $\{0,1,2\}$ is $\int_0^{1/4}\rho(y)\,dy = 1/3$). These integrals yield $P=\{1/3, 1/15, 4/15, 2/15, 1/5, 0\}$ (in the same order that the intervals were presented), from which we build the exact form of the curve $(H_q(P),C_q(P))$ for $d=3$ and $a=4$. The left panels of Fig.~\ref{fig:7} show a comparison between the numerical results for a time series of length $2^{17}$ (after dropping the initial $10^{4}$ terms) and the exact form of the $q$-complexity-entropy curve, where an excellent agreement is observed. 

We further estimate the $q$-complexity-entropy curve for other values of $a$, as shown in Fig.~\ref{fig:7}(b) for $d=4$. In these plots, we choose values of $a$ for which the map oscillates between two values ($a=3.05$), four values ($a=3.50$), and for two chaotic regimes: one ($a=3.593$) close to the onset of chaos and another at fully developed chaos ($a=4$). We note that these different regimes of the logistic map correspond to different curves. However, the values of $H_1$ and $C_1$ alone are not enough for a complete discrimination; for instance, these values are practically the same for $a=3.50$ and $a=3.593$, while the values for $q\sim0$ are very different for these two regimes. 

\begin{figure}[!ht]
\centering
\includegraphics[scale=0.33]{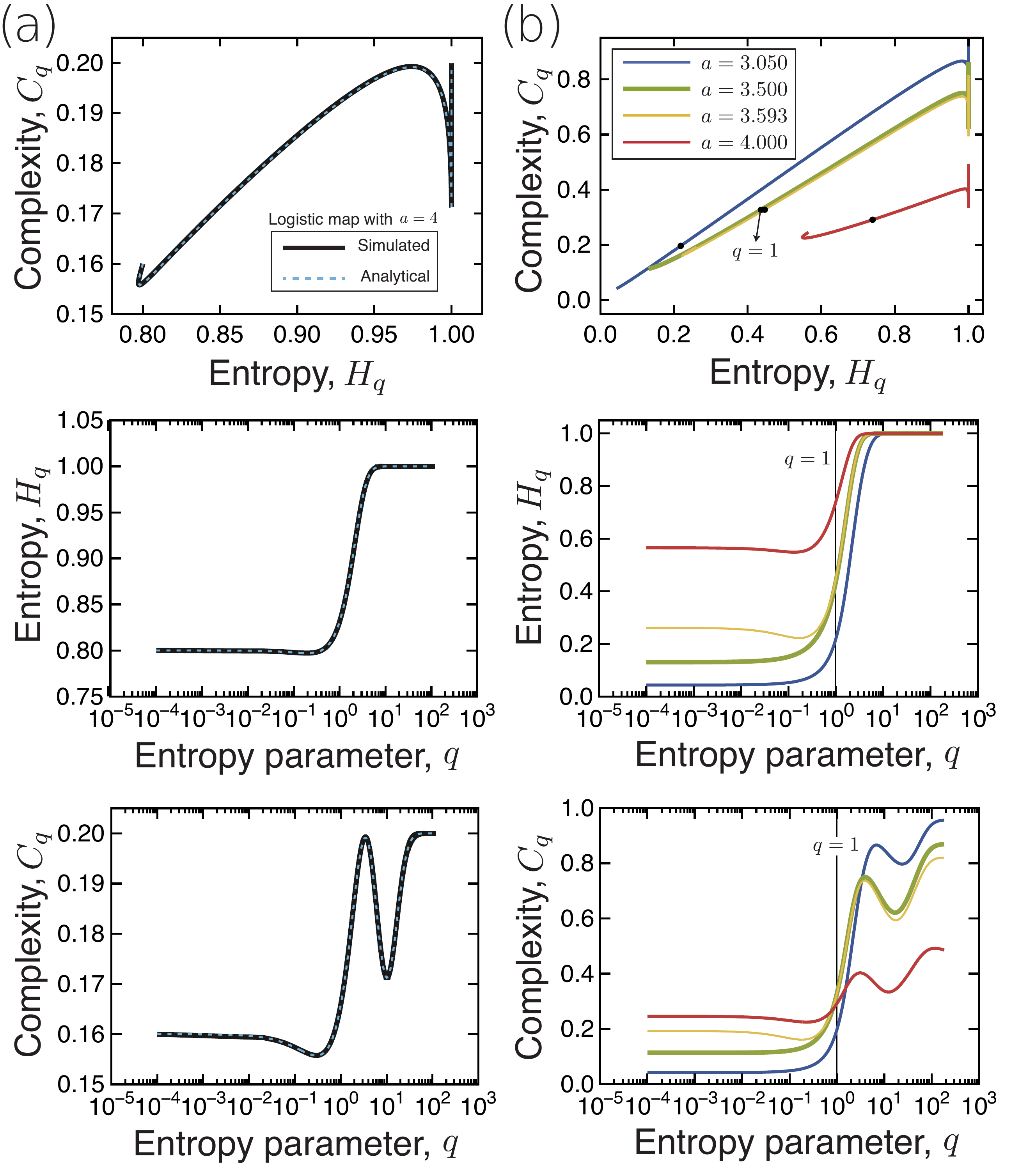}
\caption{
{Dependence of the entropy $H_q$ and complexity $C_q$ on the parameter $q$ for the logistic map.} Panel (a) shows a comparison between the values $H_q$ and $C_q$ (as well as their dependence on $q$) obtained from the simulations and the exact results for the logistic map with $a=4$ and $d=3$. Notice that a practically perfect agreement is found. Panel (b) shows the $q$-complexity-entropy curve and the dependence of $H_q$ and $C_q$ on $q$ for $d=4$ and four values of the parameter $a$: $a=3.05$ (oscillating behavior between two values), $a=3.50$ (oscillating behavior among four values), $a=3.593$ (chaotic behavior) and $a=4$ (fully developed chaos). We note that the complete differentiation among these regimes of the logistic map is only possible when considering different values of $q$. In particular, we observe that the points $(H_q,C_q)$ for $q=1$ (indicated by the black dots) are in about the same location for $a=3.50$ and $a=3.593$.
}
\label{fig:7}
\end{figure}

\section{Empirical applications}
\begin{figure}[!ht]
\centering
\includegraphics[scale=0.35]{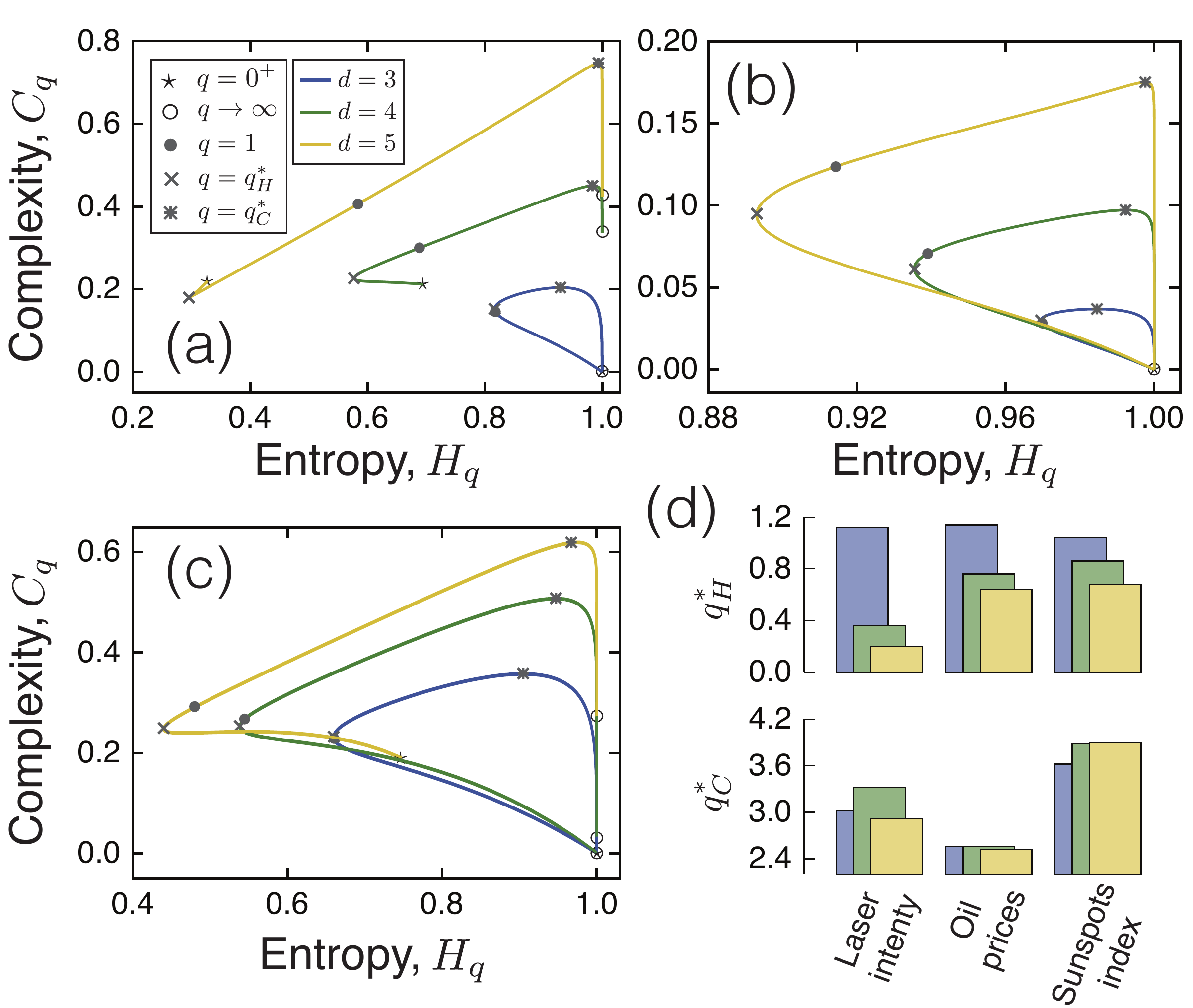}
\caption{
{Dependence of the entropy $H_q$ and complexity $C_q$ on the parameter $q$ for empirical time series.} Panel (a) shows the $q$-complexity-entropy curve for the chaotic intensity pulsations of a single-mode far-infrared NH3 laser. Panel (b) show the curves for crude oil prices (daily closing spot price of the West Texas Intermediate from January 2nd, 1986 to July 10th, 2012), and panel (c) for the monthly smoothed sunspots index (from 1749 to 2016). In all panels, the different colors refer to the embedding dimensions ($d=3,4,\text{and}~5$), the star markers indicate the points $(H_q,C_q)$ for $q=0^+$, while the open circles are the same for $q\to\infty$. Also, the gray dots indicate the points ($H_q$, $C_q$) for $q=1$, cross markers for $q=q^*_H$, and asterisk markers for $q=q^*_C$. We note that causality plane for the laser intensity is similar to those reported for chaotic maps, while the crude oil prices and sunspot index have a behavior similar to those reported for noisy time series (see Fig.~\ref{fig:6}). (d) Extreme values $q^*_H$ and $q^*_C$ obtained for each system and $d=3,4,\text{and}~5$ (different bar colors).
}
\label{fig:8}
\end{figure}

Another important test for the $q$-complexity-entropy curve is related to empirical time series. These time series usually have some degree of randomness only associated with the experimental technique employed to study a system, a feature that is well known to hinder the discrimination between experimental chaotic and stochastic signals~\cite{PhysRevE.86.046210}. Thus, in order to test the $q$-complexity-entropy curve in an experimental scenario, we first consider two empirical time series of well-known origin: the chaotic intensity pulsations of a laser~\cite{huebner1989dimensions} and the fluctuations of crude oil prices. The chaotic time series has length $n=9093$ and is freely available in Ref.~\cite{tspred}, whereas the crude oil prices refer to daily closing spot price of the West Texas Intermediate from January 2nd, 1986 to July 10th, 2012 (freely available in Ref.~\cite{eia}). The results are shown in Fig.~\ref{fig:8}(a)~and~\ref{fig:8}(b). We observe that the shape of the curve for the laser intensity is similar to those reported for chaotic maps, that is, it forms a loop only for $d=3$ (such as the Burgers map), while for higher embedding dimensions the curve is open. On the other hand, the curves for price time series form loops with a shape that resembles those of the fractional Brownian motion. We further study a time series of the monthly smoothed sunspot index, whose stochastic or chaotic nature is still debated~\cite{carbonell1993asymmetry,paluvs1999sunspot,timmer2000can,paluvs2000paluvs,mininni2000stochastic,mininni2002study,de2010fast}. By analyzing the 13-month smoothed monthly sunspot number from 1749 to 2016 ($n=3202$, freely available in Ref.~\cite{silso}), we have built the $q$-complexity-entropy curves shown in Fig.~\ref{fig:8}(c). We note that the curve is closed for $d=3$ and open for $d=4$ and $d=5$, showing a minimum value for $H_q$ for the three values of $d$; moreover, the shape of the curves are similar to those of the harmonic noise. Thus, our results suggest that the sunspot index can be described by an oscillatory behavior combined with irregularities of stochastic nature. A similar description was proposed by Mininni~\textit{et al.}~\cite{mininni2000stochastic,mininni2002study}, where a Van der Pol oscillator with a noise term was found to reproduce several features of the sunspot index. Figure~\ref{fig:8}(d) shows the values of $q$ that optimize $H_q$ and $C_q$ ($q^*_H$ and $q^*_C$) for each system. For $d=4$ and $d=5$, we note that $q^*_H$ is substantially smaller for the laser intensities than the values observed for the two other systems (which are very similar). This agrees with the results observed in Fig.~\ref{fig:6}, where we verified that $H_q$ does not have a minimum value for all embedding dimensions in the case of chaotic maps (which corresponds to $q^*_H\to0$). The price dynamics present the smallest values of $q^*_C$, followed by the laser intensities and the sunspots index (respectively), indicating that $q^*_H$ and $q^*_C$ are associated to different dynamical scales of these systems. 

Next, we test if the $q$-complexity-entropy curve can improve the discrimination of physiological signals of healthy subjects and patients with congestive heart failure. In particular, we investigate time series of the interbeat intervals from 46 healthy subjects (age~$=65.9\pm4.0$, $n=106235\pm10900$) and 15 patients (age~$=69.7\pm6.4$, $n=109031\pm12826$) with severe congestive heart failure (NYHA class III). All time series are made freely available by the PhysioNet web page~\cite{PhysioNet,PhysioNet2}. Figure~\ref{fig:9}(a) shows the average curves of all healthy subjects and patients, where loops are found for both conditions. However, the loop is broader for patients than for healthy subjects, which is compatible with the fact that the Hurst exponents of these time series are usually larger for patients than for healthy subjects~\cite{havlin1999application}. We also verify whether the values of $H_{q_H^*}$ and $C_{q_C^*}$, in comparison with $H_{1}$ and $C_{1}$, can enhance the differentiation among time series from healthy subjects and patients in a classification task. To do so, we proceed as in the fractional Brownian motion case, that is, we train a $k$-nearest neighbors algorithm in a $3$-fold cross-validation strategy. Our results show that optimized values provide a greater accuracy when compared with the usual values for $q=1$ ($\approx$80\% against $\approx$76\%), as shown in Fig.~\ref{fig:9}.

\begin{figure}[!ht]
\centering
\includegraphics[scale=0.33]{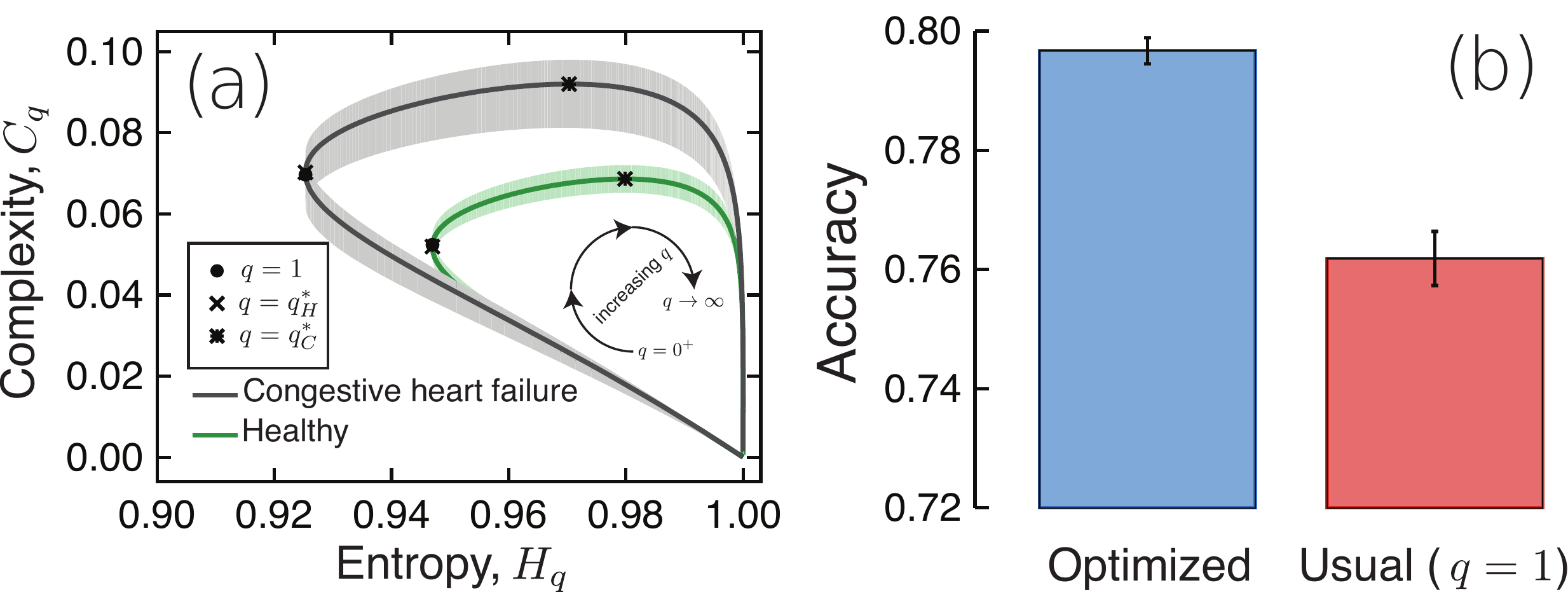}
\caption{
{Distinguishing between the interbeat intervals of healthy subjects and patients with congestive heart failure based on the values of $H_q$ and $C_q$ via the nearest neighbors algorithm.} Panel (a) shows the causality plane (for the embedding dimension $d=3$) evaluated from heart rate time series of 15 patients (age~$=69.7\pm6.4$) with severe congestive heart failure (NYHA class III, gray curve) and from the time series of 46 healthy subjects (age~$=65.9\pm4.0$, green curve). The continuous lines are the average values of $H_q$ and $C_q$ over all subjects in each group and the shaded areas are the standard error of the mean values. Panel (b) shows the accuracy of the nearest neighbors algorithm (fraction of correctly classified subjects) when employing the optimized values of $H_q$ and $C_q$ ($q=q^*_H$ and $q=q^*_C$, blue bar) and when using the values of $H_q$ and $C_q$ for $q=1$ (usual case, red bar). The error bars are 95\% confidence intervals calculated via cross-validation. We notice that the optimized values of $H_q$ and $C_q$ provide a greater accuracy when compared with the usual case ($\approx$80\% against $\approx$76\%).
}
\label{fig:9}
\end{figure}

\begin{figure}[!ht]
\centering
\includegraphics[scale=0.33]{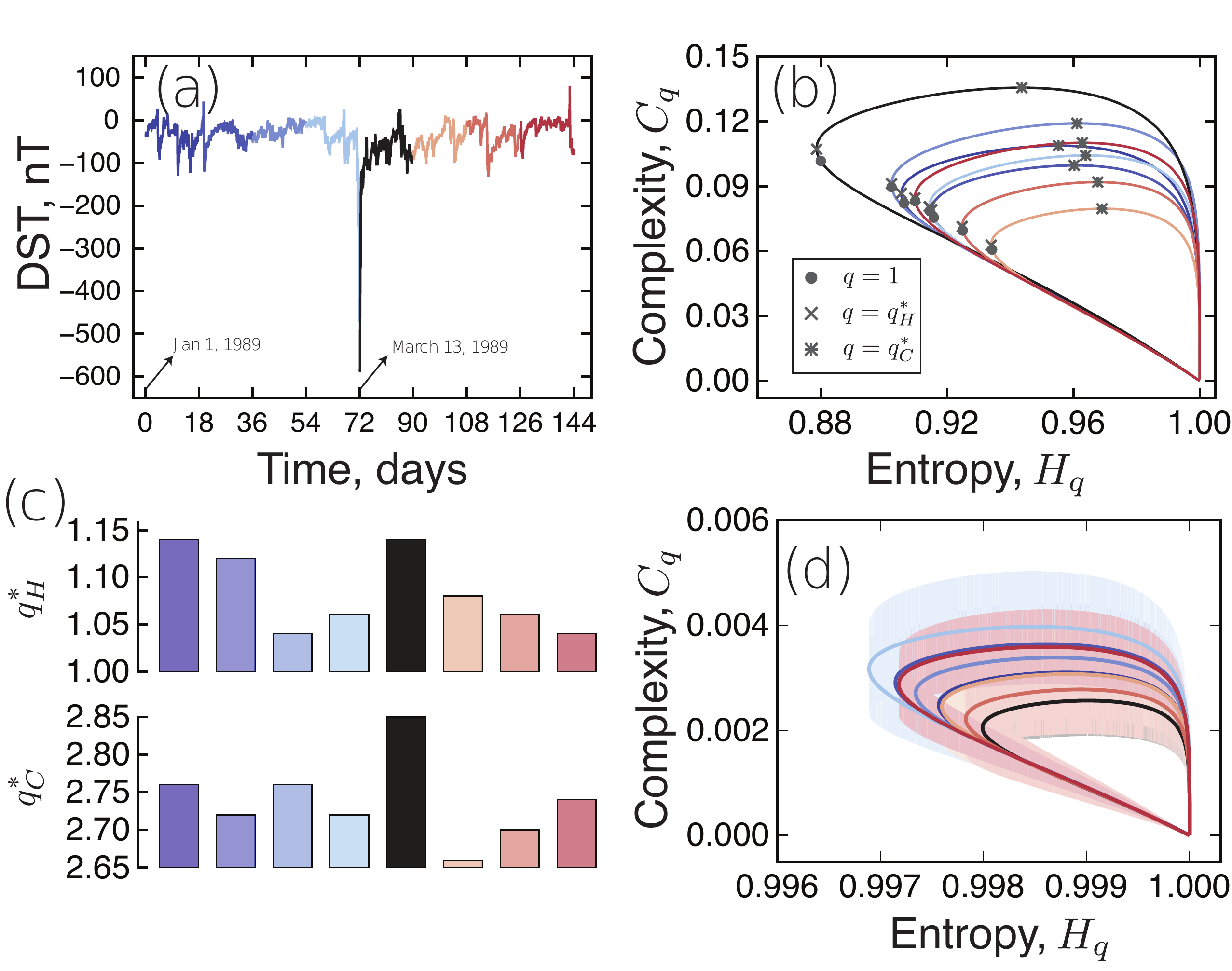}
\caption{
{Dependence of the entropy $H_q$ and complexity $C_q$ on the parameter $q$ for the Earth's magnetic activity: changes during a geomagnetic storm.} Panel (a) shows the hourly time series of the disturbance storm time index (DST, a measure of the Earth magnetic activity) from January 1st, 1989 to May 24th, 1989 (144 days). Within this period a severe geomagnetic storm struck Earth on March 13 1898, when the DST dropped to about $-600$nT. The time series is segmented in 8 periods (indicated by different colors) of 18 days and the period containing the geomagnetic storm is plotted in black. Panel (b) shows the $q$-complexity-entropy curves evaluated for each 18-days period for the embedding dimension $d=3$. The gray dots indicate the points ($H_q$, $C_q$) for $q=1$, cross markers for $q=q^*_H$, and asterisk markers for $q=q^*_C$. We notice that during the geomagnetic storm the $q$-complexity-entropy curve have the smallest value for the entropy $H_q$ and largest value for the complexity $C_q$ (that is, a broader loop). It is also worth mentioning that the period just after the storm is characterized by the shortest loop. (c) Extreme values $q^*_H$ and $q^*_C$ obtained for time series segment. (d) The $q$-complexity-entropy curves evaluated shuffled versions of each 18-days period. The continuous lines are the average value of curves over 100 realizations and the shaded areas indicate 95\% bootstrap confidence intervals. The color code employed in each plot is the same used for the DST time series.
}
\label{fig:10}
\end{figure}

Finally, as a last application, we consider a time series related to the Earth's magnetic activity: the disturbance storm time index (DST). This index reflects the average change in the Earth's magnetic field based on measurements of the equatorial ring current from a station network located along the equator on the Earth's surface. The injection of energetic ions from the solar wind into the ring current produces a magnetic field that (at the equator) is opposite to the Earth's field, often resulting in a sharp decreasing of the DST index and defining a geomagnetic storm~\cite{gonzalez1994geomagnetic}. Figure~\ref{fig:10}(a) shows the evolution of the DST index (hourly resolution) from January 1st, 1989 to May 24th, 1989 based on data freely available by the World Data Center for Geomagnetism~\cite{wdc}. During this period, a great geomagnetic storm occurred (the March 13th, 1989 geomagnetic storm~\cite{kappenman1997geomagnetic}), making the DST as lower as $-600$nT. In order to verify if the $q$-complexity-entropy curve can distinguish between the different regimes present in the DST index, we segment the data of Fig.~\ref{fig:10}(a) into time series of 18 days ($n=432$) and calculate the curves for each segment with $d=3$. Figure~\ref{fig:10}(b) shows that all curves are characterized by loops of different broadness. In particular, the period just after the beginning of the storm is characterized by the broadest loop, whereas the next data segment has the narrowest loop. We further note that, after the storm, the curve width is progressively restored to a shape similar to the one observed before the storm, reflecting the recovering dynamics of a geomagnetic storm~\cite{gonzalez1994geomagnetic}. We find that the values of $q$ that optimize $H_q$ ($q^*_H$) are close to $1$ [Fig.~\ref{fig:10}(c), upper panel] and that they are not efficient for identifying the geomagnetic storm. However, the values of $q$ that optimize $C_q$ ($q^*_C$) are very different from $1$ [Fig.~\ref{fig:10}(c), bottom panel] and capable of identifying the geomagnetic storm (notice that $q^*_C\approx 2.85$ during the geomagnetic storm, and $q^*_C<2.75$ in all other periods). We also study the $q$-complexity-entropy curves for shuffled versions of each time series segment, as shown in Fig.~\ref{fig:10}(d). After shuffling, the loops are very narrow, and no significant differences among the curves are observed. It is worth noting that the fluctuations in the values of $C_q$ are no much larger than $10^{-3}$. Assuming that the fluctuations in the original time series would have the same magnitude, this result suggests that the difference observed in Fig.~\ref{fig:10}(b) is statistically significant.

\section{Summary and Conclusions}
We have proposed an extension to the complexity-entropy causality plane of Rosso~\textit{et al.}~\cite{RossoLarrondoMartinPlastinoFuentes2007} by considering a mono-parametric generalization of the Shannon entropy (Tsallis $q$-entropy, $H_q$) and of the statistical complexity ($q$-complexity, $C_q$). Our approach for characterizing time series is based on the parametric representation of the ordered pairs $(H_q (P), C_q(P))$ on $q>0$, which we have called the $q$-complexity-entropy curve. In a series of applications involving numerically-generated and empirical time series, we have shown that the $q$-complexity-entropy curves can be very useful for characterizing and classifying time series, outperforming the original approach in several cases. In particular, the optimized version of the complexity-entropy causality plane (when using the values of $H_q^*$ and $C_q^*$) showed to be more efficient for classifying time series from the fractional Brownian motion and for distinguishing between healthy subjects and patients with congestive heart failure. These curves were also able to distinguish among different periodic behaviors of the logistic map as well as different parameters of the harmonic noise. Regarding the issue of distinguishing between chaotic and stochastic processes, we have shown that the $q$-complexity-entropy curves related to stochastic processes are usually characterized by loops, while chaotic processes display open curves, a feature that is associated with the existence of forbidden ordinal patterns in the time series.

Thus, we believe the $q$-complexity-entropy curves can be employed in a wide range of applications as a tool for characterizing time series. Naturally, other generalizations of the Shannon entropy could be employed in place of the Tsallis $q$-entropy, eventually leading to an efficient tool. One of these possibilities is the R\'{e}nyi $\alpha$-entropy~\cite{renyi}. For this case, it is possible to show that the normalized R\'{e}nyi $\alpha$-entropy (the analogous of Eq.~\ref{eq:tsallisentropynormal}) is a monotonically decreasing function of the entropic parameter $\alpha$; therefore, all $\alpha$-complexity-entropy curves will be open and the distinction between chaos and noise based on the formation of open or closed curves is not possible. This fact does not eliminate other features of the $\alpha$-complexity-entropy curves of being used for distinguishing chaos and noise as well as for classifying/characterizing time series. For instance, in a preliminary study we have observed that concavity properties of the $\alpha$-complexity-entropy curves can also be used for this task. However, a detailed study of other entropic forms and a comparison among them is outside the scope of the present work.

\acknowledgements
HVR thanks the financial support of the CNPq under Grant No. 440650/2014-3. MJ thanks the financial support of CAPES and CNPq. LZ acknowledges Consejo Nacional de Investigaciones Cient\'ificas y T\'ecnicas (CONICET), Argentina for financial support. EKL thanks the financial support of the CNPq under Grant No. 303642/2014-9.

\appendix
\section{Limiting expressions for $H_q$ and $C_q$ when $q\to0^+$ and $q\to\infty$}\label{app:Limiting_expression}

\begin{thm}
Let $P=\{p_j\}_{j=1,\ldots,d!}$ be a probability distribution, $r$ be the number of non-zero components of $P$ and $\gamma=\frac{r-1}{d!-1}$. The following statements are true:
\begin{enumerate}
\item[\upshape{(1)}] if $r=1$ then $H_q(P)=0$ and $C_q(P)=0$ for any $q>0$;
\item[\upshape{(2)}] $H_q(P)\to \gamma$ as $q\to 0^+$;
\item[\upshape{(3)}] $C_q(P)\to \gamma(1-\gamma)$ as $q\to 0^+$;
\item[\upshape{(4)}] if $r>1$ then $H_q(P)\to 1$ as $q\to\infty$;
\item[\upshape{(5)}] if $r>1$ then $C_q(P)\to 1-\gamma$ as $q\to\infty$.
\end{enumerate}
\end{thm}
\begin{proof}
(1) This follows immediately from the definition of $H_q$ and $C_q$, given in Eqs.~(\ref{eq:tsallisentropynormal}) and~(\ref{eq:qstatcomplexity}). 

(2) For any $x>0$, it is clear that $\log_q x\to x-1$ as $q\to 0^+$. Using this fact in Eq.~(\ref{eq:tsallisentropynormal}), we obtain
\begin{equation}
\begin{split}
\lim_{q\to0^+}H_q(P)&=\frac{1}{d!-1}\sum_{\substack{i=1\\p_i\ne 0}}^{d!}p_i\dpar{\frac{1}{p_i}-1}\\
&=\frac{r-1}{d!-1}\,.
\end{split}
\end{equation}

(3) It follows immediately from Eq.~(\ref{Dq*}) that \mbox{$D_q^*\to\frac{d!-1}{4d!}$} as $q\to 0^+$. From Eq.~(\ref{Dq}) we have
\begin{equation}
\begin{split}
\lim_{q\to 0^+}D_q(P,U)&=-\frac{1}{2}\sum_{\substack{i=1\\p_i\ne 0}}^{d!}p_i\dpar{\frac{p_i+1/d!}{2p_i}-1}\\
&\quad-\frac{1}{2}\sum_{i=1}^{d!}\frac{1}{d!}\dpar{\frac{p_i+1/d!}{2/d!}-1}\\
&=-\frac{1}{4}\sum_{\substack{i=1\\p_i\ne 0}}^{d!}\dpar{\frac{1}{d!}-p_i}-\frac{1}{4}\sum_{i=1}^{d!}\dpar{p_i-\frac{1}{d!}}\\
&=\frac{1}{4}\dpar{1-\frac{r}{d!}}\,.
\end{split}
\end{equation}
Using these results and item (2), we obtain from Eq.~(\ref{eq:qstatcomplexity}) that
\begin{equation}
\begin{split}
\lim_{q\to 0^+}C_q(P)&=\dpar{\frac{4d!}{d!-1}}\dpar{\frac{d!-r}{4d!}}\dpar{\frac{r-1}{d!-1}}\\
&=\dpar{1-\frac{r-1}{d!-1}}\dpar{\frac{r-1}{d!-1}}\,.
\end{split}
\end{equation}

(4) For $q>1$, Eq.~(\ref{eq:tsallisentropynormal}) can be written as
\begin{equation}
H_q(P)=\sum_{i=1}^{d!}\frac{p_i-p_i^q}{1-(d!)^{1-q}}\,.
\end{equation}
Then, if $r>1$, $H_q(P)\to \sum_{i=1}^{d\,!} p_i=1$ as $q\to\infty$.

(5) We have from Eq.~(\ref{Dq*}) that
\begin{equation}
D_q^*=\frac{K_q}{(1-q)2^{2-q}}\,,
\end{equation}
where
\begin{equation}
K_q=\frac{2^{2-q}d!-(1+d!)^{1-q}-d!(1+1/d!)^{1-q}-d!+1}{d!}\,.
\end{equation}
We note immediately that $K_q\to (1-d!)/d!$ as $q\to\infty$. We have from Eq.~(\ref{Dq}) that, for $q>1$,
\begin{equation}
\begin{split}
\frac{D_q(P,U)}{D_q^*}&=-\frac{2^{1-q}}{K_q}\left[\sum_{\substack{i=1\\p_i\ne 0}}^{d!}p_i\dpar{\frac{1}{2}+\frac{1}{2p_id!}}^{1-q}\right.\\
&\quad\left.\vphantom{\sum_{\substack{i=1\\p_i\ne 0}}^{d!}}+\sum_{i=1}^{d!}\frac{1}{d!}\dpar{\frac{p_id!}{2}+\frac{1}{2}}^{1-q}-2\right]\\
&=-\frac{1}{K_q}\left[\sum_{\substack{i=1\\p_i\ne 0}}^{d!}p_i\dpar{1+\frac{1}{p_id!}}^{1-q}\right.\\
&\quad\left.\vphantom{\sum_{\substack{i=1\\p_i\ne 0}}^{d!}}+\sum_{i=1}^{d!}\frac{1}{d!}(p_id!+1)^{1-q}-2^{2-q}\right]\,.\\
\end{split}
\end{equation}
Hence, for $r>1$,
\begin{equation}
\begin{split}
\lim_{q\to\infty}\frac{D_q(P,U)}{D_q^*}&=\dpar{\frac{d!}{d!-1}}\dpar{\frac{d!-r}{d!}}\\
&=1-\frac{r-1}{d!-1}\,.
\end{split}
\end{equation}
Therefore, using these results and item (4) on Eq.~(\ref{eq:qstatcomplexity}), $C_q(P)\to 1-\gamma$ as $q\to\infty$ whenever $r>1$.
\end{proof}

\section{Ordinal probabilities for the fractional Brownian motion}\label{app:ordinal_probabilities}

By following the results of Bandt and Shiha~\cite{BandtShiha2007}, we can write the ordinal probabilities (for the embedding dimension $d=3$) of the fractional Brownian motion with Hurst exponent $h$ as:
\begin{eqnarray}
p(\{0,1,2\}) &=& \frac{\alpha}{2}\,, \nonumber\\
p(\{0,2,1\}) &=& \frac{1-\alpha}{4}\,, \nonumber\\
p(\{1,0,2\}) &=& \frac{1-\alpha}{4}\,, \nonumber\\
p(\{2,0,1\}) &=& \frac{1-\alpha}{4}\,, \nonumber\\
p(\{1,2,0\}) &=& \frac{1-\alpha}{4}\,, \nonumber\\
p(\{2,1,0\}) &=& \frac{\alpha}{2} \nonumber\,,
\end{eqnarray}
where
\begin{equation}
\alpha = \frac{2}{\pi} \arcsin(2^{h-1})\,.
\end{equation}
Similarly, for the embedding dimension $d=4$, we have~\cite{BandtShiha2007}:
\begin{eqnarray}
p(\{0,1,2,3\}) &=& \frac{1}{8}+\frac{1}{4 \pi}\left(\arcsin\alpha_1+2 \arcsin\alpha_2 \right), \nonumber\\
p(\{0,1,3,2\}) &=& \frac{1}{8}+\frac{1}{4 \pi}\left(\arcsin\alpha_7-\arcsin\alpha_1 -\arcsin\alpha_5\right), \nonumber\\
p(\{0,2,1,3\}) &=&\frac{1}{8}+\frac{1}{4 \pi}\left(\arcsin\alpha_4-2 \arcsin\alpha_5 \right), \nonumber\\
p(\{0,2,3,1\}) &=& \frac{1}{8}+\frac{1}{4 \pi}\left(\arcsin\alpha_3+\arcsin\alpha_8 -\arcsin\alpha_5\right), \nonumber\\
p(\{0,3,1,2\}) &=& \frac{1}{8}+\frac{1}{4 \pi}\left(\arcsin\alpha_7-\arcsin\alpha_4 -\arcsin\alpha_5\right), \nonumber\\
p(\{0,3,2,1\}) &=& \frac{1}{8}+\frac{1}{4 \pi}\left(\arcsin\alpha_6-\arcsin\alpha_8 +\arcsin\alpha_2\right), \nonumber\\
p(\{1,0,2,3\}) &=& p(\{0,1,3,2\}), \nonumber\\
p(\{1,0,3,2\}) &=& \frac{1}{8}+\frac{1}{4 \pi}\left(2 \arcsin\alpha_6+\arcsin\alpha_1 \right), \nonumber
\end{eqnarray}
\begin{eqnarray}
p(\{1,2,0,3\}) &=& p(\{0,3,1,2\}), \nonumber\\
p(\{1,2,3,0\}) &=& p(\{0,3,2,1\}), \nonumber\\
p(\{1,3,0,2\}) &=& p(\{0,2,3,1\}), \nonumber\\
p(\{1,3,2,0\}) &=& p(\{0,2,3,1\}), \nonumber\\
p(\{2,0,1,3\}) &=& p(\{0,2,3,1\}), \nonumber\\
p(\{2,0,3,1\}) &=& p(\{0,3,2,1\}), \nonumber\\
p(\{2,1,0,3\}) &=& p(\{0,3,2,1\}), \nonumber\\
p(\{2,1,3,0\}) &=& p(\{0,3,1,2\}), \nonumber\\
p(\{2,3,0,1\}) &=& p(\{1,0,3,2\}), \nonumber\\
p(\{2,3,1,0\}) &=& p(\{0,1,3,2\}), \nonumber\\
p(\{3,0,1,2\}) &=& p(\{0,3,2,1\}), \nonumber\\
p(\{3,0,2,1\}) &=& p(\{0,3,1,2\}), \nonumber\\
p(\{3,1,0,2\}) &=& p(\{0,2,3,1\}), \nonumber\\
p(\{3,1,2,0\}) &=& p(\{0,2,1,3\}), \nonumber\\
p(\{3,2,0,1\}) &=& p(\{0,1,3,2\}), \nonumber\\
p(\{3,2,1,0\}) &=& p(\{0,1,2,3\}),\nonumber
\end{eqnarray}
where
\begin{eqnarray}
\alpha_1 &=& \frac{1+3^{2 h}-2^{2 h +1}}{2},\nonumber\\
\alpha_2 &=& 2^{2 h -1}-1,\nonumber\\
\alpha_3 &=& \frac{1-3^{2 h}-2^{2 h}}{2\times6^{h} },\nonumber\\
\alpha_4 &=& \frac{3^{2 h}-1}{2^{2 h +1}},\nonumber\\
\alpha_5 &=& 2^{h -1},\nonumber\\
\alpha_6 &=& \frac{2^{2 h}-3^{2 h}-1}{2\times3^{h}},\nonumber\\
\alpha_7 &=& \frac{3^{2 h}-2^{2 h}-1}{2^{h+1}},\nonumber\\
\alpha_8 &=& \frac{2^{2 h}-1}{3^{h}}.\nonumber
\end{eqnarray}
By using these values, we find the exact form of the distribution $P=\{p(\pi_j)\}_{j=1,\ldots,d!}$ and the $q$-complexity-entropy curve $(H_q(P),C_q(P))$.

\section{Definition of the eight chaotic maps employed in our study}\label{app:chaotic_maps}

The Burgers map is defined as
\begin{equation*}
\begin{split}
x_{k+1}&= a x_{k} - y_{k}^2\\
y_{k+1}&= b y_{k} + x_{k} y_{k}
\end{split}\,,
\end{equation*}
and we have chosen $a=0.75$ and $b=1.75$. The time series that we have analyzed is $(x_{k}+y_{k})^2$ with $x_0=-0.1$ and $y_0=0.1$.

The cubic map is defined as
\begin{equation*}
\begin{split}
x_{k+1}&= a x_k(1-x_k^2)
\end{split}\,,
\end{equation*}
and we have chosen $a=3$ and $x_0=0.1$.

The Gingerbreadman map is defined as
\begin{equation*}
\begin{split}
x_{k+1}&= 1 - y_k + |x_k|\\
y_{k+1}&= y_{k}
\end{split}\,,
\end{equation*}
and we have chosen $x_0=0.5$ and $y_0=3.7$. The time series that we have analyzed is $(x_{k}+y_{k})^2$.

The logistic map is defined as
\begin{equation*}
\begin{split}
x_{k+1}&= a x_k(1-x_k)
\end{split}\,,
\end{equation*}
and we have chosen $a=4$ and $x_0=0.1$. 

The H\'enon map is defined as
\begin{equation*}
\begin{split}
x_{k+1}&= 1 - a x_k^2 + y_k\\
y_{k+1}&= b x_{k}
\end{split}\,,
\end{equation*}
and we have chosen $a=1.4$ and $b=0.3$. The time series that we have analyzed is $(x_{k}+y_{k})^2$ with $x_0=0$ and $y_0=0.9$.

The Ricker map is defined as
\begin{equation*}
\begin{split}
x_{k+1}&= a x_k\,{ }e^{-x_k} 
\end{split}\,,
\end{equation*}
and we have chosen $a=20$ and $x_0=0.1$. 

The sine map is defined as
\begin{equation*}
\begin{split}
x_{k+1}&= a \sin(\pi x_k)
\end{split}\,,
\end{equation*}
and we have chosen $a=1$ and $x_0=0.1$. 

The Tinkerbell map is defined as
\begin{equation*}
\begin{split}
x_{k+1}&= x_k^2 - y_k^2 + a x_k + b y_k \\
y_{k+1}&= 2x_k y_k + c x_k + d y_k
\end{split}\,,
\end{equation*}
and we have chosen $a=0.9$, $b=-0.6$, $c=2.0$, and $d=0.5$. The time series that we have analyzed is $(x_{k}+y_{k})^2$ with $x_0=-0.1$ and $y_0=0.1$. 

\bibliography{complex_plane_q}
\bibliographystyle{apsrev4-1}

\end{document}